\documentclass[12pt]{article}
\usepackage{fullpage}
\usepackage[latin1]{inputenc}
\usepackage[T1]{fontenc}
\usepackage{ae,aecompl}
\usepackage{amsmath,amsfonts}

\usepackage{booktabs}

\newenvironment{pf}{\noindent\textbf{Proof.}\quad}{\hfill{$\Box$}}
\usepackage[dvips]{graphicx}
\usepackage[usenames,dvipsnames]{color}
\usepackage{epsfig}
\usepackage{rotating}
\usepackage{amssymb}
\usepackage[latin1]{inputenc}
\usepackage{verbatim}
\newtheorem{lem}{Lemma}
\newtheorem{thm}{Theorem}

\newtheorem{conj}{Conjecture}
\newtheorem{cor}{Corollary}
\newtheorem{ex}{Example}
\newcommand{\F}{\mathbb{F}}
\newcommand{\Z}{\mathbb{Z}}

\newcommand{\beg}{\begin{equation}}
\newcommand{\eeg}{\end{equation}}

\newcommand{\m}{\mbox}
\newcommand{\mf}{\hspace{5mm} \mbox{ }}
\newcommand{\mz}{\hspace{2mm} \mbox{ }}
\newcommand{\bite}{\begin{itemize}}
\newcommand{\eite}{\end{itemize}}
\newcommand{\al}{\alpha}

\newcommand{\cev}{\overset{{}_{{\bf\leftarrow}}}}
\newenvironment{proof}{{\em Proof:}}
\def\bra#1{\mathinner{\langle{#1}|}}
\def\ket#1{\mathinner{|{#1}\rangle}}

\newcommand{\ketbra}[2]{|#1\rangle\langle#2|}

\begin{document}

\title{Mixed graph states}
\author{
Constanza Riera, Ramij Rahaman, Matthew G. Parker
\thanks{C. Riera is with Bergen University College, Norway. E-mail: csr@hib.no.\
M.G. Parker is with University of Bergen, Norway. E-mail: matthew@ii.uib.no.\ R. Rahaman is with  University of Allahabad, India. E-mail: ramijrahaman@gmail.com}
}

\date{\today}
\maketitle

\section{Introduction}

The quantum {\em graph state}\cite{graphstate} is a pure quantum $n$-qubit state whose {\em joint stabilizer} can be written
using a symmetric matrix whose elements are Pauli matrices. The Pauli matrices are:
$$ I = \left ( \begin{array}{rr} 1 & 0 \\ 0 & 1 \end{array} \right ), \mf
X = \left ( \begin{array}{rr} 0 & 1 \\ 1 & 0 \end{array} \right ), \mf
Z = \left ( \begin{array}{rr} 1 & 0 \\ 0 & -1 \end{array} \right ), \mf
Y = iXZ, $$
where $i = \sqrt{-1}$.
More precisely, in \cite{graphstate}, graph states in $n$ qubits are defined as the unique, common eigenvector in $({\mathbb{C}}^2)^n$ to the set of independent commuting observables:
$$K_a=X_aZ_{{\cal{N}}_a}=X_a\prod_{b\in{\cal{N}}_a} Z_b$$
where $A_a=I\otimes I\otimes\cdots\otimes A\otimes I\otimes\cdots\otimes I$, ${\cal{N}}_a$ are the neighbours of $a$ in the associated graph (see below), and where the eigenvalues to the observables $K_a$ are +1 for all $a=0,\ldots,n-1$.
%Specifically, and without loss of
%generality, it has been shown that we can put $X$ and/or $Y$ on the diagonal and $Z$ and/or $I$ off the diagonal. 
A natural generalisation is given by $$K_a=\sigma_aZ_{{\cal{N}}_a}=\sigma_a\prod_{b\in{\cal{N}}_a} Z_b$$
where $\sigma\in\{X,Y\}$. For 
the {\em mixed graph states} of this paper we require this extra generality.
%as $U_v=X_vZ_{{\cal{N}}_v}$.
%However, for  expand this to $U_v=M_vZ_{{\cal{N}}_v}$, where $M=X$ or $Y$. The need for this %definition is that in our definition of {\em mixed graph states} (see next section) we will %encounter such stabilizers.

We represent a pure quantum graph state by a simple graph, hence the name {\em graph state}.
For instance, for $n = 3$, the simple graph with vertices $V = \{0,1,2\}$ and edges $E = \{01,02\}$ has an
adjacency matrix $A = \left ( \begin{small}
\begin{array}{lll} 0 & 1 & 1 \\ 1 & 0 & 0 \\ 1 & 0 & 0 \end{array}
\end{small} \right )$ and can be used to represent the pure quantum state of 3 qubits that is
jointly stabilized (with eigenvalue 1) by $X \otimes Z \otimes Z$, $Z \otimes X \otimes I$,
and $Z \otimes I \otimes X$. This set of $n = 3$ joint stabilizers can be written in matrix form as
${\cal A} = i^{I_3 \star A} \star X^{I_3} \star Z^A =
\left ( \begin{small}
\begin{array}{lll} 1 & 1 & 1 \\ 1 & 1 & 1 \\ 1 & 1 & 1 \end{array}
\end{small} \right ) \star \left ( \begin{small}
\begin{array}{lll} X & I & I \\ I & X & I \\ I & I & X \end{array}
\end{small} \right ) \star \left ( \begin{small}
\begin{array}{lll} I & Z & Z \\ Z & I & I \\ Z & I & I \end{array}
\end{small} \right ) =  \left ( \begin{small}
\begin{array}{lll} X & Z & Z \\ Z & X & I \\ Z & I & X \end{array}
\end{small} \right )$, where `$\star$ is the element-wise product,
$I_3$ is the $3 \times 3$ identity, and
the connection with the adjacency matrix, $A$, should be clear. As another example, for
$A = \left ( \begin{small}
\begin{array}{lll} 0 & 1 & 1 \\ 1 & 1 & 1 \\ 1 & 1 & 1 \end{array}
\end{small} \right )$, we get ${\cal A} = i^{I_3 \star A} \star X^{I_3} \star Z^A =
\left ( \begin{small}
\begin{array}{lll} 1 & 1 & 1 \\ 1 & i & 1 \\ 1 & 1 & i \end{array}
\end{small} \right ) \star \left ( \begin{small}
\begin{array}{lll} X & I & I \\ I & X & I \\ I & I & X \end{array}
\end{small} \right ) \star \left ( \begin{small}
\begin{array}{lll} I & Z & Z \\ Z & Z & Z \\ Z & Z & Z \end{array}
\end{small} \right ) =  \left ( \begin{small}
\begin{array}{lll} X & Z & Z \\ Z & Y & Z \\ Z & Z & Y \end{array}
\end{small} \right )$.

As the system is stabilized by any of the $n$ matrix rows, it is therefore also stabilized by
any product of the rows. In general, this set of stabilizers forms a stabilizer code with $2^n$
elements, i.e. a stabilizer group, $S$. For
instance, for $A = \left ( \begin{small}
\begin{array}{lll} 0 & 1 & 1 \\ 1 & 0 & 0 \\ 1 & 0 & 0 \end{array}
\end{small} \right )$, we have that
$(X \otimes Z \otimes Z)(Z \otimes X \otimes I) = (Y \otimes Y \otimes I)$
is also a joint stabilizer of our graph state, as is
$(Z \otimes I \otimes X)(X \otimes Z \otimes Z)(Z \otimes X \otimes I) = -(X \otimes Y \otimes Y)$.

\vspace{2mm}

Given the matrix ${\cal A}$ of stabilizers the implicit assumption above
is that these elements are stabilizing a
{\em pure} quantum state, i.e. that the $n$ qubits are not entangled with an environment. Given
this assumption then the rows of ${\cal A}$ must commute, otherwise they cannot
define (stabilize) an $n$-qubit pure state. Let $\ket{\psi}$ be this pure state and let
$U$ be any row of ${\cal A}$. Then $U\ket{\psi}=\ket{\psi}$, where $\ket{\psi}$ is unique, the density matrix of the graph state is
$\rho=\ket{\psi}\bra{\psi}$, and $U\rho U^\dagger=\rho$.

%We have that $U_v\ket{\psi}=\ket{\psi}$, where $\ket{\psi}$ is unique. The density matrix of the %graph state is then
%$\rho=\ket{\psi}\bra{\psi}$. We have then that $U_v\rho U_v^\dagger=\rho$ for all $v$.

\begin{lem} Let $\rho=\ket{\psi}\bra{\psi}$ be the density matrix of a pure graph state,
%\begin{footnote}{If $\rho$ is not pure, then it is not uniquely determined as, for instance, %$X\otimes Z$, $Z\otimes X$  will stabilize $\left|(-1)^{x_0x_1}\right>\left<(-1)^{x_0x_1}\right|$ %but also %$1/2(\left|(-1)^{x_0x_1}\right>\left<(-1)^{x_0x_1}\right|+\left|(-1)^{x_0x_1+x_0+x_1}\right>\left<(-1)^{x_0x_1+x_0+x_1}\right|)$, %and $1/4(\left|(-1)^{x_0x_1}\right>\left<(-1)^{x_0x_1}\right|+
%\left|(-1)^{x_0x_1+x_0+x_1}\right>\left<(-1)^{x_0x_1+x_0+x_1}\right|+
%\left|(-1)^{x_0x_1+x_0}\right>\left<(-1)^{x_0x_1+x_0}\right|+
%\left|(-1)^{x_0x_1+x_1}\right>\left<(-1)^{x_0x_1+x_1}\right|)$.}\end{footnote}
and $U$ as above. Then, as $U$ is taken over all rows of ${\cal A}$,  then $U\rho U^\dagger=\rho$ determines $\rho$ uniquely up to local Pauli unitary equivalence. \end{lem}
\begin{pf} %Since $U$ is unitary, $U^\dagger=U^{-1}$. Furthermore, if $U\in\{X,Z,I\}^{\otimes n}$, then since %$X^{-1}=X$ and $Z^{-1}=Z$, then $U^{-1}=U$. PROOF IN PROGRESS...
Since $U$ is a tensor product of local Pauli matrices, $U\ket{\psi}\bra{\psi} U^\dagger=\ket{\psi}\bra{\psi}$ implies $U\ket{\psi}=\pm\ket{\psi}$. This means that $\ket{\psi}$ is a common eigenvector of the set of matrices, $U$, with eigenvalue $\pm1$. Let $P$ be an $n$-fold tensor product of Pauli matrices. Then it is straightforward that
$UP\rho P^\dagger U^\dagger= PU\rho U^\dagger P^\dagger = P \rho P^\dagger$, so $\ket{\psi}$ is defined by $U$ up to local Pauli unitary equivalence.
\end{pf}

\vspace{2mm}

The symmetric matrix ${\cal A}$ is completely characterised, up to, but not including, commutation,
by ${\cal A}_4 = A + \omega I_n$, where ${\cal A}_4$ is an additive matrix over $\F_4$, $I_n$ is the $n \times n$ identity matrix, $\omega$ is a primitive generator of $\F_4$, i.e.
$\omega^3 = 1$, and where $I \rightarrow 0$,
$Z \rightarrow 1$, $X \rightarrow \omega$, $Y \rightarrow \omega^2$. For instance, for $A = \left ( \begin{small}
\begin{array}{lll} 0 & 1 & 1 \\ 1 & 1 & 1 \\ 1 & 1 & 1 \end{array}
\end{small} \right )$,
${\cal A}_4 = \left ( \begin{small}
\begin{array}{lll} \omega & 1 & 1 \\ 1 & \omega^2 & 1 \\ 1 & 1 & \omega^2 \end{array}
\end{small} \right )$. The set of $2^n$ stabilizers then become the set of $2^n$ codewords of an
$\F_4$ additive code (where multiplication of rows of ${\cal A}$ becomes $\F_4$ addition of rows of
${\cal A}_4$). This code is self-dual with respect to the Hermitian inner product because
the matrix is symmetric with $\omega$ or $\omega^2$ on the diagonal.
Observe, however, that this map from stabilizers to $\F_4$ does
not take into account matrix commutation as it does not
capture the minus sign of, for instance, $XZ = -ZX$, as this translates to $w.1 = 1.w$.

\vspace{3mm}

In this paper we define and characterise mixed graph states, by considering square matrices ${\cal A}$ that are
not, in general, symmetric. In previous
work \cite{LEdirect} the resulting $\F_4$-additive codes (not, in general, self-dual) have been
classified and referred to as {\em directed graph codes}. For instance, the adjacency matrix
$A = \left ( \begin{small}
\begin{array}{lll} 0 & 1 & 0 \\ 0 & 0 & 1 \\ 0 & 1 & 0 \end{array}
\end{small} \right )$
is not symmetric and describes a mixed graph
\begin{footnote}{A {\em mixed graph} is a graph where some of the edges may be directed.
In contrast, all edges of a {\em directed graph} are directed. %In (cite previous paper) we referred
%to mixed graphs as directed graphs, but in fact the codes that we constructed there
%should have been referred to as {\em mixed graph codes} as they were obtained, more generally, from
%mixed graphs. In this paper, mixed graphs include as special cases, the situations where no edges
%are directed, and where all edges are directed.
}\end{footnote}
comprising an undirected edge $12$, and
a directed edge $0 \rightarrow 1$.

As before we can map such a matrix to a matrix of stabilizers,
thus:

$$ {\cal A} = \left ( \begin{small}
\begin{array}{lll} X & Z & I \\ I & X & Z \\ I & Z & X \end{array}
\end{small} \right ). $$

But what is a natural quantum interpretation of such a matrix? Rows 0 and 2, and 1 and 2 commute, but
rows 0 and 1 anti-commute. A quantum interpretation is only possible if all rows pairwise commute, so we choose to
extend the rows by extra columns until the rows pairwise commute. Implicit in this column extension choice is the assumption that we add environmental qubits as opposed to, more generally, qudits or continuous variables.

\vspace{2mm}

In this paper we show how such matrices can define a class of mixed states, which we call
{\em mixed graph states}. The products of the rows of the matrix still form a group, $S$, up to
a global constant, but now
the group does not have a joint eigenvector, i.e. we don't have a pure state, since some of the rows anti-commute.

Due to anti-commutativity,
some of the elements of the group are only defined up to a global constant of $\pm 1$ (there is
no natural ordering on the rows of ${\cal A}$).
For instance, for the matrix ${\cal A}$ associated with our example above,
$(X \otimes Z \otimes I)(I \otimes X \otimes Z) = -(I \otimes X \otimes Z)(X \otimes Z \otimes I)$, as the two rows anti-commute. So $S$ contains $\pm i(X \otimes Y \otimes Z)$.

We can than define a density matrix, $\rho$, which is stabilized by all members of $S$. Specifically,
$$ s\rho s^{\dag} = \rho \mf \mf \forall s \in S. $$
More generally, $\rho$ is stabilized by any $\al s$, where $|\al| = 1$, which is why we don't have to concern ourselves with global phase constants in members of $S$.
%Mixed graphs (that allows directed and undirected edges) however are related to non-commuting stabilizers and will not be associated with a pure state. However, we can associate them with a mixed quantum state:
%We will then represent the arrowed edge $u\rightarrow v$ locally by the non-commuting stabilizers $X_uZ_v$, $I_uX_v$, while the non-arrowed edge $u-v$ is represented by the commuting stabilizers $X_uZ_v$, $Z_uX_v$.  We define the mixed graph state associated with the graph $G$ as   the quantum state given by $\rho$ such that $U\rho U^\dagger=\rho\ \forall U\in S$, where $S$ is the stabilizer group generated by the  $n$ unitary matrices  $X_jZ_{\vec{{\cal{N}}}_j},\,j=0,\ldots,n-1$, where ${\vec{\cal{N}}}_j$ are the neighbours of $j$ in the graph.

\vspace{3mm}

Throughout this paper we define the pure state vectors using a Boolean function notation
\cite{thesis}. Specifically, we associate with our $n$-qubit graph state
an $n$-variable homogeneous generalised quadratic Boolean function $p : \F_2^n \rightarrow {\mathbb{Z}}_4$,
where
$p(x_0,x_1,\ldots,x_{n-1})=\sum_{j<k} 2A_{jk}x_jx_k + \sum_j A_{jj}x_j$, and $A$ is the modified adjacency matrix of our graph state with elements $A_{jk}$, defined as $A_{jk}=\left\{\begin{array}{lc}
\Gamma_{jk}&,j\neq k\\
0&,j= k \mbox{ and }\sigma_j=X\\
1&,j= k \mbox{ and } \sigma_j=Y
\end{array}\right.$,
with $\Gamma$ the adjacency matrix of the graph. Then one can write the graph state $\ket{\psi}$ as $\ket{\psi}= 2^{-\frac{n}{2}}i^p$.

In later sections, we refer to the nodes where $\sigma_j=X$ as white nodes, and to the nodes where $\sigma_j=Y$ as red nodes.

\vspace{3mm}

\begin{ex}
%{\em Example 1:}
We show here the mixed graph state associated with the directed triangle:
%\begin{figure}
\begin{center}
\scalebox{0.5}{\begin{picture}(0,0)%
\includegraphics{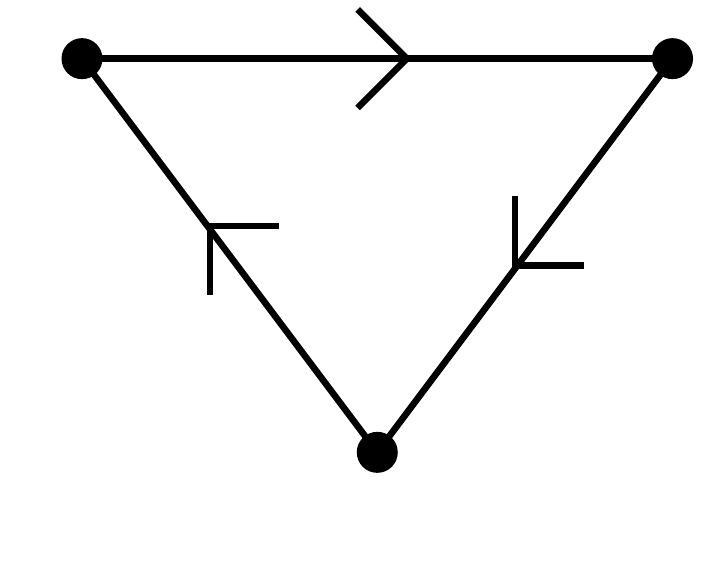}%
\end{picture}%
\setlength{\unitlength}{4144sp}%
\begingroup\makeatletter\ifx\SetFigFont\undefined%
\gdef\SetFigFont#1#2#3#4#5{%
  \reset@font\fontsize{#1}{#2pt}%
  \fontfamily{#3}\fontseries{#4}\fontshape{#5}%
  \selectfont}%
\fi\endgroup%
\begin{picture}(3270,2577)(2776,-2380)
\put(2791,-151){\makebox(0,0)[lb]{\smash{{\SetFigFont{20}{24.0}{\familydefault}{\mddefault}{\updefault}{\color[rgb]{0,0,0}$0$}%
}}}}
\put(6031,-151){\makebox(0,0)[lb]{\smash{{\SetFigFont{20}{24.0}{\familydefault}{\mddefault}{\updefault}{\color[rgb]{0,0,0}$1$}%
}}}}
\put(4411,-2266){\makebox(0,0)[lb]{\smash{{\SetFigFont{20}{24.0}{\familydefault}{\mddefault}{\updefault}{\color[rgb]{0,0,0}$2$}%
}}}}
\end{picture}%
}
\end{center}
%\end{figure}
\end{ex}
 The stabilizer group $S$ is generated, up to $\pm 1$ constants, by:
$$\begin{array}{c}
%S\mbox{ generated by:}\\
X\otimes Z\otimes I\\
I\otimes X\otimes Z\\
Z\otimes I\otimes X\end{array}$$
Because some of the rows of ${\cal A}$ anti-commute, instead of commute, there is no pure quantum state $\left|\psi\right>$ such that $(X\otimes Z\otimes I)\left|\psi\right>=(I\otimes X\otimes Z)\left|\psi\right>=(Z\otimes I\otimes X)\left|\psi\right>=\left|\psi\right>$.

Instead we can associate $S$ with the mixed graph state whose density matrix is:
%\column{.40\textwidth}
$$ \rho_0=%\frac{1}{2}\left|(-1)^{x_0x_1}\right>\left<(-1)^{x_0x_1}\right|+\frac{1}{2}\left|(-1)^{x_0x_1+x_1}\right>\left<(-1)^{x_0x_1+x_1}\right|=
\frac{1}{8}\left(\begin{array}{cccccccc}
1&0&0&i&1&0&0&-i\\
0&1&i&0&0&-1&i&0\\
0&-i&1&0&0&i&1&0\\
-i&0&0&1&-i&0&0&-1\\
1&0&0&i&1&0&0&-i\\
0&-1&-i&0&0&1&-i&0\\
0&-i&1&0&0&i&1&0\\
i&0&0&-1&i&0&0&1
\end{array}\right).$$
We will discuss in later sections how we obtain this density matrix.
It is easy to check that 
$$ (X\otimes Z\otimes I)\rho_0(X\otimes Z\otimes I)^\dagger=
(I\otimes X\otimes Z)\rho_0(I\otimes X\otimes Z)^\dagger=
(Z\otimes I\otimes X)\rho_0(Z\otimes I\otimes X)^\dagger=\rho_0. $$

%(X\otimes Z\otimes I)\rho_0(X\otimes Z\otimes I)^\dagger=\rho_0\\
%(I\otimes X\otimes Z)\rho_0(I\otimes X\otimes Z)^\dagger=\rho_0\\
%(Z\otimes I\otimes X)\rho_0(Z\otimes I\otimes X)^\dagger=\rho_0\end{array}$$

$\rho_0$ is not the unique mixed graph state associated with $S$, however: Other density matrices $\rho_j$ such that
$$(X\otimes Z\otimes I)\rho_j(X\otimes Z\otimes I)^\dagger=
(I\otimes X\otimes Z)\rho_j(I\otimes X\otimes Z)^\dagger=
(Z\otimes I\otimes X)\rho_j(Z\otimes I\otimes X)^\dagger=\rho_j$$
are: %such that $(X\otimes Z)\rho(X\otimes Z)^\dagger=\rho$, and $(I\otimes X)\rho(I\otimes X)^\dagger=\rho$?
 $$\rho_1=\frac{1}{8}\left(\begin{small}\begin{array}{cccccccc}
 1&  0&1&0&0&i&0&-i  \\
    0&1&0&1&-i&0&i&0\\
    1&0&1&0&0&i&0&-i\\
    0&1&0&1&-i&0&i&0\\
    0&i&0&i&1&0&-1&0\\
    -i&0&-i&0&0&1&0&-1\\
    0&-i&0&-i&-1&0&1&0\\
    i&0&i&0&0&-1&0&1
      \end{array}\end{small}\right), \rho_2=\frac{1}{8}\left(\begin{small}\begin{array}{cccccccc}
1&1&0&0&0&0&i&-i\\
1&1&0&0&0&0&i&-i \\
0&0&1&-1&i&i&0&0\\
 0&0&-1&1&-i&-i&0&0\\
0&0&-i&i&1&1&0&0\\
0&0&-i&i&1&1&0&0\\
-i&-i&0&0&0&0&1&-1\\
i&i&0&0&0&0&-1&1  \end{array}\end{small}\right)$$

There are also another 3 matrices which are locally equivalent  to these 3  (see lemma \ref{equiv}). Thus we define
   the mixed graph state associated to the graph as $$\rho=\sum_{i=0}^5c_j\rho_j$$ where $c_j\geq0,\ \sum_{j=0}^5c_j=1$, i.e. the convex sum of these
  density matrices. Note that, although only at most 3 of them are locally inequivalent, we include them all in the general sum $\rho$, since changing $\rho_j$ for a  locally equivalent $\rho_j'$ does not necessarily give local equivalence in $\rho$.

  Each of the $\rho_j$ can be associated with a pure graph state in 4 qubits. For instance
  $$\rho_0=\frac{1}{2}\left|i^{2(x_0x_2+x_1x_2)+x_2}\right>\left<i^{2(x_0x_2+x_1x_2)+x_2}\right|+\frac{1}{2}\left|i^{2(x_0x_2+x_1x_2+x_1+x_2)+x_2}\right>\left<i^{2(x_0x_2+x_1x_2+x_1+x_2)+x_2}\right|$$ is the density matrix resulting from a measurement of qubit 3 of the pure graph state $i^{2(x_0x_2+x_1x_2+x_1x_3+x_2x_3)+x_2}$ (in terms of Boolean functions, this would correspond to the evaluation of $x_3$ at 0 and 1). Similarly, %(generalised\footnote{See next section.})
  $\rho_1$ and $\rho_2$ are the density matrices resulting from a measurement of qubit 3 of the pure graph states $i^{2(x_0x_1+x_0x_3+x_2x_3+x_0)+x_2}$ and $i^{2(x_0x_1+x_1x_2+x_0x_3+x_1x_3)+x_1}$, while the remaining $\rho_j$ will be given by the addition of linear terms to the polynomial.

\section{Mixed graph states}
\label{mixeddef}

To recap, a matrix such as ${\cal A} = \left ( \begin{small}
\begin{array}{lll} X & Z & I \\ I & X & Z \\ I & Z & X \end{array}
\end{small} \right )$ has rows that are not fully pairwise
commuting, but can still be interpreted as a quantum object
by making it part of a larger fully commuting matrix, i.e. where we choose the environment
appropriately, and this will imply that our quantum object is a mixed state, in fact a {\em mixed graph
state} (in this paper, we limit our study to the addition of environmental qubits as opposed to, more generally, qudits).

Let $G$ be the mixed graph defined by adjacency matrix $A$, and $G_b$ the undirected graph that has adjacency matrix $\Gamma=A+A^T$. Thus
$G_b$ is the simple graph obtained from $G$ by erasing all undirected edges of $G$ and making all directed edges of $G$ undirected.

What is the minimum number of columns, $e$, necessary to append to
${\cal A}$ so as to make all rows pairwise commute? We can then add the same number, $e$, of rows to
the matrix so as to make it square again, making sure that all resultant rows pairwise commute with the previous rows and with each other.
For instance, we can extend our example by $e = 1$ columns and rows to get
${\cal A}' = \left ( \begin{small}
\begin{array}{cccc}
\multicolumn{3}{c}{{\cal A}} & \begin{array}{c} Z \\ X \\ I \end{array} \\
Z & I & I & X
\end{array}
\end{small} \right ) =
\left ( \begin{small} \begin{array}{llll}
X & Z & I & Z \\ I & X & Z & X \\ I & Z & X & I \\ Z & I & I & X
\end{array}
\end{small} \right )$, which is fully commuting.
So for this example the minimum number of columns and rows necessary to make the matrix fully commuting
was $e = 1$. Finally, as this matrix is fully commuting we can,
by suitable row multiplications, recover its graph form
${\cal A}^e = \left ( \begin{small} \begin{array}{llll}
X & Z & I & Z \\ Z & X & Z & I \\ I & Z & X & I \\
Z & I & I & X \end{array} \end{small} \right )$.

We consider the pure graph state described by ${\cal A}^e$ to be a {\em parent} graph state, being
a parent of the mixed graph state described by ${\cal A}$. We then denote the resultant mixed state given by tracing out the environmental qubits as its {\em child}.

%Notation: In graph theory, if an edge $e=(u,v)$ is directed from $u$ to $v$, $u$ is called a {\em directed predecessor} of $v$, and $v$ a {\em directed successor} of $u$.
%Note that $(v,u)$ might also be an edge in the mixed graph, in which case the edge is undirected %(not arrowed) edge $u- v$) or not be an edge (in which case we have an arrowed edge $u\rightarrow %v$);
%The set of direct successors of $u$ is referred to  as the out-neighbourhood of $u$, while the set of direct predecessors of $u$ is referred to  as the in-neighbourhood of $u$.
% In this paper,
%we will denote the out-neighbourhood of $u$ by $\vec{{\cal{N}}}_u$, and the in-neighbourhood of $u$ by %$\overset{{}_{\leftarrow}}{{\cal{N}}_u}$
%$\cev{{\cal{N}}_u}$. We define  ${\cal{N}}_u:=\vec{{\cal{N}}}_u\cap\cev{{\cal{N}}_u}$, the undirected neighbourhood of $u$.

%\vspace{2mm}

For our example, ${\cal A} = \left ( \begin{small}
\begin{array}{lll} X & Z & I \\ I & X & Z \\ I & Z & X \end{array}
\end{small} \right )$, we have
$A = \left ( \begin{small}
\begin{array}{lll} 0 & 1 & 0 \\ 0 & 0 & 1 \\ 0 & 1 & 0 \end{array}
\end{small} \right )$, and the minimum number of columns and rows necessary to
add to ${\cal A}$ to make it fully commuting is
$$e = \frac{1}{2}{\m{rank}}(\Gamma) = \frac{1}{2}{\m{rank}}
\left ( \begin{small}
\begin{array}{lll} 0 & 1 & 0 \\ 1 & 0 & 0 \\ 0 & 0 & 0 \end{array}
\end{small} \right ) = 1$$

In general:

\begin{lem} [special case of \cite{Brun}]
The minimum number, $e$, of columns and rows required to be added to
${\cal A}$ to make its rows pairwise fully commuting is given by
\beg e = \frac{1}{2}{\m{rank}}(\Gamma). \label{dirrank} \eeg
We refer to $e$ as the {\em mixed rank}.
\label{rankres}
\end{lem}
%\begin{pf}
%(ADD REFERENCE TO Quantum LDPC work by Brun .....etc..).
%\end{pf}

For our example we extended to a 4-qubit
pure graph parent state where the 4th qubit is part of the environment. This pure graph state can be written
using Boolean function notation as
$$ \ket{\psi}_e = \frac{1}{4}(-1)^{x_0x_1 + x_0x_3 + x_1x_2}
= \frac{1}{4}(1,1,1,-1,1,1,-1,1,1,-1,1,1,1,-1,-1,-1)^T. $$
Tracing out the 4th qubit, means summing the projectors obtained by fixing $x_3 = 0$ and $1$,
respectively. We obtain the projectors $\ketbra{\phi_0}{\phi_0}$ and
$\ketbra{\phi_1}{\phi_1}$, where
$\ket{\phi_0} = \frac{1}{\sqrt{8}}(-1)^{x_0x_1 + x_1x_2}$,
$\ket{\phi_1} = \frac{1}{\sqrt{8}}(-1)^{x_0x_1 + x_1x_2 + x_0}$, and
$$ \begin{array}{ll}
\rho & = \ketbra{\phi_0}{\phi_0} +  \ketbra{\phi_1}{\phi_1} \\
     &  =
\frac{1}{8}\left ( \begin{tiny} \begin{array}{rrrrrrrr}
1 & 1 & 1 & -1 & 1 & 1 & -1 & 1 \\
1 & 1 & 1 & -1 & 1 & 1 & -1 & 1 \\
1 & 1 & 1 & -1 & 1 & 1 & -1 & 1 \\
-1& -1& -1&  1 &-1 &-1 &  1 &-1 \\
1 & 1 & 1 & -1 & 1 & 1 & -1 & 1 \\
1 & 1 & 1 & -1 & 1 & 1 & -1 & 1 \\
-1& -1& -1&  1 &-1 &-1 &  1 &-1 \\
1 & 1 & 1 & -1 & 1 & 1 & -1 & 1
\end{array} \end{tiny} \right ) +
\frac{1}{8}\left ( \begin{tiny} \begin{array}{rrrrrrrr}
1 &-1 & 1 &  1 & 1 &-1 & -1 &-1 \\
-1& 1 &-1 & -1 &-1 & 1 &  1 & 1 \\
1 &-1 & 1 &  1 & 1 &-1 & -1 &-1 \\
1 &-1 & 1 &  1 & 1 &-1 & -1 &-1 \\
1 &-1 & 1 &  1 & 1 &-1 & -1 &-1 \\
-1& 1 &-1 & -1 &-1 & 1 &  1 & 1 \\
-1& 1 &-1 & -1 &-1 & 1 &  1 & 1 \\
-1& 1 &-1 & -1 &-1 & 1 &  1 & 1
\end{array} \end{tiny} \right ) \\
 & = \frac{1}{4}\left ( \begin{tiny} \begin{array}{rrrrrrrr}
1 & 0 & 1 & 0 & 1 & 0 & -1 & 0 \\
0 & 1 & 0 & -1& 0 & 1 & 0  & 1 \\
1 & 0 & 1 & 0 & 1 & 0 & -1 & 0 \\
0 &-1 & 0 & 1 & 0 & -1& 0  & -1 \\
1 & 0 & 1 & 0 & 1 & 0 & -1 & 0 \\
0 & 1 & 0 & -1& 0 & 1 & 0  & 1 \\
-1& 0 & -1& 0 & -1& 0 & 1  & 0 \\
0 & 1 & 0 & -1& 0 & 1 & 0  & 1
\end{array}  \end{tiny} \right ).
\end{array} $$

So we can interpret the stabilizer matrix ${\cal A}$ as representing the mixed quantum state $\rho$.
However this is not a complete interpretation for two reasons:
\bite
    \item One can extend by more than one column/row and still obtain a fully commuting matrix.
    \item The matrix ${\cal A}$ can be extended in multiple ways, leading to multiple, possibly
    inequivalent, mixed states
    $\rho$.
\eite

We shall, initially, avoid the first issue by stipulating that we only extend by the minimum possible number,
$e$, of columns/rows, as given by (\ref{dirrank}). This is compatible with the pure graph state
formulation as $e = 0$ forces the state to be pure. As an example of the second issue, observe that
it is equally valid to extend ${\cal A}$ to
${\cal A}' = \left ( \begin{small}
\begin{array}{cccc}
\multicolumn{3}{c}{{\cal A}} & \begin{array}{c} X \\ Z \\ I \end{array} \\
I & Z & I & X
\end{array}
\end{small} \right ) =
\left ( \begin{small} \begin{array}{llll}
X & Z & I & X \\ I & X & Z & Z \\ I & Z & X & I \\ I & Z & I & X
\end{array}
\end{small} \right )$
from which, by multiplicative row operations, we obtain
${\cal A}^e = \left ( \begin{small} \begin{array}{llll}
X & I & I & I \\ I & X & Z & Z \\ I & Z & X & I \\
I & Z & I & X \end{array} \end{small} \right )$. It is clear (since the former is a connected graph and the latter is not) that the parent graph state
described by this ${\cal A}^e$ is locally inequivalent to the previous parent. Therefore
the resultant mixed state,
$\rho = \ketbra{\phi_0}{\phi_0} +  \ketbra{\phi_1}{\phi_1}$, obtained by tracing out the 4th qubit,
is different from the previous mixed state. In general we shall have multiple parent graph states leading
to multiple density matrices, so our overall density matrix shall be a convex
sum of these density matrices.

\vspace{2mm}

Define the {\em codespace} (i.e. the stabilizer group $S$) of the $n \times n$ matrix, ${\cal A}$, to be the $2^n$ stabilizers formed by products of
one or more of the rows of ${\cal A}$ - remember that, as some of the rows of ${\cal A}$
anti-commute, some of the members of $S$ are only defined up to a global multiplicative
constant of $\pm 1$. Likewise, the codespace of ${\cal A}^e$ comprises
$2^{n+e}$ stabilizers, but now all rows of ${\cal A}^e$ commute so the global constant is always $1$.

\vspace{2mm}

We refer to $H = \frac{1}{\sqrt{2}}\left ( \begin{array}{rr} 1 & 1 \\ 1 & -1 \end{array} \right )$
as the {\em Hadamard} matrix, and to
$N = \frac{1}{\sqrt{2}}\left ( \begin{array}{rr} 1 & i \\ 1 & -i \end{array} \right )$ as the
{\em negaHadamard} matrix. Both are unitary.

\begin{lem}
The Pauli group, $\{I,X,Z,Y\}$, is conjugated (up to a factor of $\pm1$) by the group generated by $\{I,H,N\}$. Specifically,
$$ \begin{array}{r||r|r|r|r|r|r}
   & I & H & N & N^2 & NH & HN \\ \hline
 X & X & Z &-Y & Z   & X  & Y  \\
 Z & Z & X & X & -Y  & -Y & Z \\
 Y & Y & -Y& -Z& -X  &  Z & -X
\end{array} $$
\label{ConjP}
\end{lem}
For instance, $X = HZH^{\dag}$, $-Z = NYN^{\dag}$, and $-X = N^2Y(N^2)^{\dag}$. Lemma
\ref{ConjP} implies that, by conjugation, to within a factor of $\pm 1$, we can permute the elements of the
stabilizer matrix, column-wise. There are $3!=6$ such permutations.

The following is well known (proof omitted):
\begin{lem}
Consider a quantum system comprising a local system, $L$, possibly entangled with an environment, $E$, and
let $\rho_L$ be the density matrix that defines the local system. Then unitary conjugation,
$U_E(L \times E)U_E^{\dag}$, on the environment leaves $\rho_L$ unchanged.
\label{envuni}
\end{lem}
%\begin{pf} (omitted)
%Obviously well-known and fundamental but a proof would be good.
%See for instance \cite{course}.
%\end{pf}

\vspace{3mm}
\begin{lem} The extension columns and rows do not depend on the direction of the arrows, although the parent graphs depend on the direction of the arrows.

\end{lem}
\begin{pf}
The extension columns and rows depend exclusively on the commutativity or anti-commutativity of the stabiliser basis, which is not dependent on arrow direction.
\end{pf}

\section{General Density Matrix Stabilized by ${\cal A}$}\label{rho}

%One can always write a general density matrix of dimension $d$ in the form
%\beg \rho = \frac{1}{d}\left ( I + \Lambda \right ) =
% \frac{1}{d}\left ( I + \sum_{j=1}^{d^2-1} \al_j\lambda_j \right ),
%\label{gendense}
%\eeg
%where the $\lambda_j$ are traceless, orthogonal, and Hermitian, i.e. they are physical observables,
%and where $\al_j = \langle \rho \lambda_j \rangle = \Tr ( \rho \lambda_j )$, and
%$\sum_{j=1}^{d^2-1} |\al_j|^2 \le 1$, i.e. $|\al_j|^2$ is the probability of measuring $\lambda_j$.

%\vspace{3mm}

A density matrix of $n$ qubits can be written in a Pauli basis, i.e. where the basis elements
are tensor products of elements from $\{I,X,Y,Z\}$, i.e.
\beg \begin{array}{ll}
\rho = \sum_{k \in \Z_4^n} a_k {\tilde{\sigma_k}}, & \m{ where }
k = (k_0,k_1,\ldots,k_{n-1}) \in \Z_4^n, {\tilde{\sigma_k}} = \bigotimes_{j=0}^{n-1} \sigma_{k_j}, \\
     & \sigma_0 = I, \sigma_1 = X, \sigma_2 = Y, \sigma_3 = Z, \m{ and } \sum_{k \in \Z_4^n} |a_k|^2 \le 1,
\end{array}
\label{PauliBasis}
\eeg
where the $a_k$ are real. Necessary and sufficient conditions for $\rho$ to be a density matrix, i.e.
derived from a statistical sum of pure states by tracing out the environment, are:
\bite
    \item $\rho$ is Hermitian (i.e. equal to its transpose conjugate).
    \item $\rho$ is positive semi-definite.
    \item $\m{Tr}(\rho) = 1$.
\eite

%\vspace{2mm}

%Using (\ref{gendense}) we can re-express (\ref{PauliBasis}) as follows.

%\beg \begin{array}{l}
%\rho = \frac{1}{2^n} \left ( I^{\otimes n}
%  + \sum_{k \in \Z_4^n, k \ne \bf{0}} \al_k{\tilde{\sigma_k}} \right ),  \mf
%  \sum_{k \in \Z_4^n, k \ne \bf{0}} |\al_k|^2 \le 1.
%\end{array}
%\label{PauliBasisI}
%\eeg

%\vspace{3mm}

We refer to ${\bf P} = \{{\tilde{\sigma}}_k \mz | \mz k \in \Z_4^n \}$, as the set of Pauli codewords (stabilizers) of size $|{\bf P}| = 4^n$.
%Let $U_0 \otimes U_1 \otimes \ldots \otimes U_{n-1}$, $U_j \in \{I,X,Y,Z\}$ be referred to as a %length-$n$ Pauli codeword (or stabilizer). There are $4^n$ such codewords.
Let ${\cal A}$ be the $n \times n$ matrix for a mixed graph state. Then ${\cal A}^T$ is
the matrix for the mixed graph state obtained from the graph for
${\cal A}$ by reversing all arrows. As before, we say that ${\cal A}$ generates a Pauli subgroup, $S \subset {\bf P}$. We will
now show that ${\cal A}^T$ generates the dual group of $S$, namely $S^{\perp} \subset {\bf P}$.

\begin{thm}
The set of length-$n$ Pauli
words (stabilizers) that commute with all rows of ${\cal A}$ forms a multiplicative group, $S^{\perp}$, of size $2^n$, being precisely the words generated by ${\cal A}^T$, to within $\pm 1$ constants.
\label{revmix}
\end{thm}
\begin{pf} It is straightforward 
%note to self: should we include proof? it's easy, but still...
 to show that every row of ${\cal A}^T$ commutes with
every row of ${\cal A}$. To show that there are no more, observe that
each row of ${\cal A}^T$ can be viewed as a restriction. Only $4^n/2^1$
(i.e. half) of all possible length-$n$ Pauli words commute with a given row. The rows of
${\cal A}^T$ are independent so the $n$ rows of ${\cal A}^T$ jointly
commute with only $4^n/2^n = 2^n$ Pauli words. But the rows of ${\cal A}^T$ generate
$2^n$ words, which must be distinct and, therefore, must be all of them.
\end{pf}

See section \ref{AppA} for a discussion of the $\pm 1$ coefficients.
\vspace{1cm}

Graph states can be also described as the sum of its stabilisers, as seen in \cite{graphstates}. In order to state a similar result for mixed graph state, we need the following lemmas:

By the support of a matrix, $M$, we mean the set of non-zero element positions of $M$, i.e. $\m{support}(M) = \{(i,j) | M_{i,j} \ne 0\}$.

\begin{lem}\label{supp}The elements of the stabilizer group, $S^{\perp}$, of a mixed graph have non-intersecting support.
\label{supp}
\end{lem}
\begin{pf} By inspection $X$ and $Y$ have the same support, as do $I$ and $Z$. But the support of $X$ and $Y$ is non-intersecting
with that of $I$ and $Z$. Moreover if two matrices ${\tilde{\sigma}}_k$ and ${\tilde{\sigma}}_{k'}$ are both tensor products of Pauli matrices, and if they differ
in at least one tensor position, where one matrix has $X$ or $Y$, and the other matrix has $I$ or $Z$, then ${\tilde{\sigma}}_k$ and ${\tilde{\sigma}}_{k'}$ must also
have non-intersecting support.
The matrix ${\cal A}$ only has $X$ and/or $Y$ on the diagonal and $I$ and/or $Z$ off the diagonal.
%, so any two rows of ${\cal A}$, being tensor products of Pauli matrices, must be non-intersecting as any two rows must be %non-intersecting in precisely two tensor element positions.
Let ${\cal A}$ have rows numbered $0$ to $n-1$. Let $R$ and $R'$ be two subsets of
$\{0,1,\ldots,n-1\}$ and consider the matrices, ${\tilde{\sigma}}_{k_R}$ and
${\tilde{\sigma}}_{k_{R'}}$, being the product of the rows of ${\cal A}$ indexed by
$R$ and $R'$, respectively. Then ${\tilde{\sigma}}_{k_R}$ and ${\tilde{\sigma}}_{k_{R'}}$ have $X$ or $Y$ at tensor positions indexed by $R$ and $R'$, respectively, and $I$ or $Z$ at all other positions. So, unless $R$ and $R'$ are equal, they must be non-intersecting in at least one tensor position, so therefore
${\tilde{\sigma}}_{k_R}$ and ${\tilde{\sigma}}_{k_{R'}}$ have non-intersecting support. The elements of the stabilizer group of the mixed graph are obtained by ranging $R$ over
all subsets, i.e. over all stabilizer codewords. Therefore any two of them must be non-intersecting.
\end{pf}

\vspace{2mm}

For a given stabilizer group $S$, the condition $s\rho s^\dagger=\rho\ \forall s\in S$ is equivalent to the condition $s\rho=\rho s\ \forall s\in S$. The Pauli basis description of (\ref{PauliBasis}) imposes the condition that $\rho=\sum_k a_k\tilde{\sigma}_k$,
where the $\tilde{\sigma}_k$ are tensor products of Pauli matrices, such that $s\tilde{\sigma}_k =\tilde{\sigma}_k s\ \forall s\in S$. In other words, $\rho$ is the sum of tensors of Pauli matrices that commute with all $s\in S$. By theorem \ref{revmix}, the Pauli subgroup of matrices that commute with all elements in the stabilizer group $S$ is the stabilizer group, $S^{\perp}$, of the mixed graph with the arrows reversed. The intersection of $S$ and $S^{\perp}$ is then the row space of ${\cal A}$ associated with the vertices of the mixed graph, $G$, that are isolated in $G_b$. But what choices do we have for the angle of $a_k$ in (\ref{PauliBasis}), such that $\rho$ is also Hermitian (as a density matrix must be)? Let $a_k = \hat{a}_k\overline{a}_k$ such that $|\hat{a}_k| = 1$ and
$\overline{a}_k$ is real.
Lemma \ref{supp} implies that each component $a_k\tilde{\sigma}_k$ of $\rho$ must, itself, be Hermitian, so this fixes $\hat{a}_k$ precisely, to within $\pm 1$. For example, consider
${\cal A}^{\perp} = \left ( \begin{array}{cc} X & Z \\ I & X \end{array} \right )$. Then $S^{\perp}$ contains ${\tilde{\sigma}}_{12} = \pm i(X \otimes Y) =
\pm \left ( \begin{array}{rrrr} 0 & 0 & 0 & 1 \\ 0 & 0 & 1 & 0 \\ 0 & -1 & 0 & 0 \\ -1 & 0 & 0 & 0
\end{array} \right )$, which is not Hermitian, so we choose $\hat{a}_{12} = i$ or
$\hat{a}_{12} = -i$ in this case. In the sequel we deal with subgroups, $S$ and $S^{\perp}$, of the tensored Pauli group, of size $2^n$ (to within
$\pm 1$ coefficients). So instead of indexing by $k \in \Z_4^n$, we index using $K  \subset \Z_n =  \{0,1,\ldots,n-1\}$ relative to the rows of the group generating matrices ${\cal A}$ and ${\cal A}^T$, respectively. Specifically,
$$ \rho = \sum_{k \in \Z_4^n} a_k\tilde{\sigma}_k = \frac{1}{2^n}\sum_{K \subset \Z_n} b_K{\cal A}^T_K, \mf \sum_{K \subset \Z_n} |b_K|^2 \le 1, \mf {\cal A}^T_K \in S^{\perp}, $$
%\sum_{j \in \F_2^n} b_j{\cev{s}_j}, \mf \sum_j |b_j|^2 \le 1, \mf {\cev{s}_j} \in S^{\perp}, $$
where the $b_K \in \{\pm 1,\pm i\}$, and ${\cal A}^T_K = \prod_{h \in K} {\cal A}^T_h$,
%$j = (j_0,j_1,\ldots,j_{n-1})$, and ${\cev{s}_j} = \pm \prod_{h=0}^{n-1} j_h{\cal A}^T_h$,
where ${\cal A}^T_h$ is the $h$th row of ${\cal A}^T$, and $\prod_{h \in \{j,j'\}} {\cal A}^T_h = {\cal A}^T_{j'}{\cal A}^T_j$ if $j' > j$, i.e. we fix the product ordering with lowest indices on the right.
Moreover, for $V$ obtained from $A^T$, $b_K$ is unique if the elements in $K$ are pairwise commutative, and is otherwise only defined up to $\pm 1$. For instance, if $K = \{2,4,5\}$ then
$\prod_{h \in K} {\cal A}^T_h = {\cal A}^T_5{\cal A}^T_4{\cal A}^T_2$, and $b_K$  is unique only if the elements in $\{2,4,5\}$ are pairwise commutative.

\section{The maximum commutative subgroups of $S^\perp$}

We now show that any commutative subgroup of the group of operators generated by the rows of ${\cal{A}}^\perp = {\cal{A}}^T$ is contained in a commutative subgroup of maximum size
$2^{n-e}$. %In later sections,We then conjecture that any such maximum commutative subgroup is a basis for a density matrix arising from the partial trace of a pure parent graph state.
Firstly, we show how the commutativity of the Pauli operators of a mixed graph state can be expressed in terms of the associated simple undirected graph, $G_b$, and its corresponding adjacency matrix, $\Gamma = \Gamma_{G_b}$.
%This way, the problem of finding a commutative subgroup reduces to a binary linear code problem.
Then, in section \ref{sectionchildren}, we show that the density matrix of any child of a pure graph state parent can be represented as a weighted sum of the terms of a maximal commutative subgroup of $S^\perp$.
\vspace{2mm}

%Given a mixed graph $G$ and one of its  parents,  we know how to define the density matrix of its child in terms of the Pauli group, by theorem \ref{pauli sum}, and that the terms of the sum form a commutative subgroup of $S^\perp$;  we will show in this section that we can find all commutative subgroups it directly from the graph $G$, or alternatively from the adjacency matrix of $G_b$  (the simple undirected graph with adjacency matrix $\Gamma={\cal{A}}+{\cal{A}}^T$). %We know by the previous section that the density matrices corresponding to a mixed graph are given by maximal size commutative subgroups of the stabilizer group of the dual mixed graph (the ones with all arrows reversed), and this size is  $2^{n-e}$, and more generallty commutative subgroups of the stabilizer group of the dual mixed graph of any size (up to the maximal size, of course).

%Note that this is a generalisation of the problem of finding the members of $S^{\perp}$ associated with a $+1$ coefficient, given in section \ref{rho}.

% If $u$ is a neighbour of $v$ in $G_b$, it means that the corresponding rows of the stabilizer anti-commute. We know there exist extensions of the stabilizer of $G$ of minimal size $e=\frac{\mbox{rank}(\Gamma)}{2}$, to the stabilizer of a pure graph state. From theorem \ref{pauli sum}, we know that these extensions will each yield a commutative sugbroup of the stabilizer group of the dual mixed graph. We get these commutative sugbroups directly from the graph:

\begin{lem}\label{easyextension}
Let ${\cal{A}}_j^\perp$ be row $j$ of ${\cal A}^\perp$.
Let $J,K \subset \{0,1,\ldots,n-1\}$.
%\footnote{NB: notation is different from the one we used in the previous section.}.
Then ${\cal{A}}^\perp_j$ commutes with ${\cal{A}}^\perp_K$
%$\prod_{k\in K} {\cal{A}}^\perp_k$
iff the number of elements of $K\cap {\cal{N}}_j$ is even. Furthermore, ${\cal{A}}^\perp_J$
%$\prod_{j\in J}{\cal{A}}^\perp_j$
commutes with ${\cal{A}}^\perp_K$
%$\prod_{k\in K}{\cal{A}}^\perp_k$
iff the sum over $J$ of the number of elements of $K\cap {\cal{N}}_j$ is even.
\end{lem}

\begin{pf} ${\cal{A}}^\perp_j$ anticommutes with any ${\cal{A}}^\perp_k$ such that $k\in K\cap {\cal{N}}_j$, that is, any $k$ in the neighbourhood of $j$, thus
${\cal{A}}^\perp_j{\cal{A}}^\perp_k = -{\cal{A}}^\perp_k{\cal{A}}^\perp_j$. If there are an even number of neighbours then the minus signs cancel each other out so that ${\cal{A}}^\perp_j$ commutes with ${\cal{A}}^\perp_K$, otherwise ${\cal{A}}^\perp_j$ anticommutes with ${\cal{A}}^\perp_K$.
More generally,  ${\cal{A}}^\perp_J$ commutes with ${\cal{A}}^\perp_K$ iff, after passing, for all $j \in J$, ${\cal{A}}^\perp_j$ through ${\cal{A}}^\perp_K$, an even number of minus signs are generated. otherwise ${\cal{A}}^\perp_J$ anticommutes with ${\cal{A}}^\perp_K$.
\end{pf}

\begin{ex} Let $G_b$ be the simple graph 01,04,12,23,34,14. Then ${\cal{A}}^\perp_0$ commutes with ${\cal{A}}^\perp_3$ (since they are independent, so the intersection is empty), and both commute with ${\cal{A}}^\perp_K = {\cal{A}}^\perp_{\{1,2,4\}} = {\cal{A}}^\perp_4{\cal{A}}^\perp_2{\cal{A}}^\perp_1$, since 0 has 1 and 4 as neighbours, and 3 has 2 and 4 as neighbours, so both have two elements in the intersection with $K=\{1,2,4\}$. Note that e.g. ${\cal{A}}^\perp_0({\cal{A}}^\perp_4{\cal{A}}^\perp_2{\cal{A}}^\perp_1)=-{\cal{A}}^\perp_4{\cal{A}}^\perp_0{\cal{A}}^\perp_2{\cal{A}}^\perp_1=-{\cal{A}}^\perp_4{\cal{A}}^\perp_2{\cal{A}}^\perp_0{\cal{A}}^\perp_1=+{\cal{A}}^\perp_4{\cal{A}}^\perp_2{\cal{A}}^\perp_1{\cal{A}}^\perp_0$.
We also have that ${\cal{A}}^\perp_3{\cal{A}}^\perp_2$ commutes with ${\cal{A}}^\perp_4{\cal{A}}^\perp_1{\cal{A}}^\perp_0$, since both 2 and 3 have one neighbour in $K=\{0,1,4\}$, so the sum over $J=\{2,3\}$ is even. %Note that $(c_3c_2)(c_4c_1c_0)=c_3c_4c_2c_1c_0=-c_3c_4c_1c_2c_0=-c_3c_4c_1c_0c_2=+c_4c_3c_1c_0c_2=c_4c_1c_3c_0c_2=c_4c_1c_3c_0c_2=c_4c_1c_0c_3c_2$.
\end{ex}

\vspace{2mm}

Let $\Gamma_{k,j}$ be the $(k,j)$th element of $\Gamma = A + A^T$.
\begin{lem}
\label{comm}
${\cal{A}}^\perp_j$ commutes with ${\cal{A}}^\perp_K$
%$\prod_{k\in K} {{\cal{A}}^\perp_k}$
iff $\bigoplus_{k\in K} \Gamma_{k,j}=0$, where $'\oplus'$ is the binary sum. Furthermore, ${\cal{A}}^\perp_J$ commutes with ${\cal{A}}^\perp_K$ iff $\bigoplus_{j\in J}\bigoplus_{k\in K} \Gamma_{k,j}=0$.
Let $v_K = (v_0,v_1,\ldots,v_{n-1}) \in F_2^n$ be such that $v_k = 1$ if $k \in K$ and $v_k = 0$ otherwise. Then the property that ${\cal{A}}^\perp_J$ commutes with ${\cal{A}}^\perp_K$ translates to the condition $v_K\Gamma v_J^T = 0$.
\end{lem}

\begin{pf}
We first prove that ${{\cal{A}}^\perp_j}$ commutes with ${\cal{A}}^\perp_K$ iff $\bigoplus_K \Gamma_{k,j}=0$:
%Due to the definition of $G_b$, $c_j$ commutes with $c_k$ iff $c_{k,j}=0$.
By lemma \ref{easyextension}, ${{\cal{A}}^\perp_j}$ commutes with ${\cal{A}}^\perp_K$ $\Leftrightarrow$ the number of elements of $K\cap {\cal{N}}_j$ is even, i.e. there are an even number of rows $k\in K$ with $\Gamma_{k,j}=1$, i.e. $\bigoplus_K \Gamma_{k,j}=0$.
Now we prove that ${\cal{A}}^\perp_J$ commutes with ${\cal{A}}^\perp_K$ iff $\bigoplus_{j\in J}\bigoplus_{k\in K} \Gamma_{k,j}=0$:
 % We have seen that each $c_j$ commutes with $\prod_{k\in K} c_k$ iff $\bigoplus_K r_{k,j}=0$
By lemma \ref{easyextension}, ${\cal{A}}^\perp_J$ commutes with ${\cal{A}}^\perp_K$ iff the sum over $J$ of the number of elements of $K\cap {\cal{N}}_j$ is even, i.e. there are an even number of rows $k\in K$ such that $\Gamma_{k,j}=1$, i.e. $\bigoplus_{j\in J}\bigoplus_{k \in K} \Gamma_{k,j}=0$.
Finally, note that $v_K\Gamma = (\bigoplus_{k \in K} \Gamma_{k,0},\bigoplus_{k \in K} \Gamma_{k,1},\ldots,\bigoplus_{k \in K} \Gamma_{k,n-1})$, from which it follows that
$v_K\Gamma v_J^T = \bigoplus_{j \in J}\bigoplus_{k \in K} \Gamma_{k,j}$.
\end{pf}

As $\Gamma$ is symmetric, we get that $v_K\Gamma v_J^T = v_J\Gamma v_K^T  = 0$, which ensures that commutativity is a symmetric relationship. %so if $\prod_{j\in A}c_j$ commutes with $\prod_{k\in K} c_k$, then $\prod_{k\in K}c_k$ commutes with $\prod_{j\in A} c_j$, as it should do.

\begin{ex}
For $\Gamma=\left(\begin{array}{ccccc}
0&1&0&0&1\\
1&0&1&0&1\\
0&1&0&1&0\\
0&0&1&0&1\\
1&1&0&1&0
\end{array}\right)$, ${{\cal{A}}^\perp_0}$ commutes with ${{\cal{A}}^\perp_3}$ since $\Gamma_{3,0}=0$, and both ${{\cal{A}}^\perp_0}$ and ${{\cal{A}}^\perp_3}$ commute with ${{\cal{A}}^\perp_4}{{\cal{A}}^\perp_2}{{\cal{A}}^\perp_1}$, since $\bigoplus_{\{1,2,4\}}\Gamma_{k,0}=1\oplus 0\oplus1=0$ and $\bigoplus_{\{1,2,4\}}\Gamma_{k,3}=0\oplus 1\oplus1=0$.
We also have that ${{\cal{A}}^\perp_3}{{\cal{A}}^\perp_2}$ commutes with ${{\cal{A}}^\perp_4}{{\cal{A}}^\perp_1}{{\cal{A}}^\perp_0}$, since $\bigoplus_{j\in \{2,3\}}\bigoplus_{k\in \{0,1,4\}} \Gamma_{k,j}=(0\oplus 1\oplus 0)\oplus (0\oplus 0\oplus 1)=0$, i.e. $v_K\Gamma v_J^T = 0$, for $v_K = (1,1,0,0,1)$ and $v_J = (0,0,1,1,0)$.
\end{ex}

If ${\cal{A}}^\perp_J$ commutes with ${\cal{A}}^\perp_K$ and ${\cal{A}}^\perp_{J'}$ commutes with ${\cal{A}}^\perp_K$, and  $J\cap J'=\emptyset$, then
${\cal{A}}^\perp_J{\cal{A}}^\perp_{J'}$
%$\prod_{j\in J,J'}{\cal{A}}^\perp_j$
commutes with ${\cal{A}}^\perp_K$. Similarly, if ${\cal{A}}^\perp_J$ anti-commutes with ${\cal{A}}^\perp_K$ and ${\cal{A}}^\perp_{J'}$ anti-commutes with ${\cal{A}}^\perp_K$, and  $J\cap J'=\emptyset$, then ${\cal{A}}^\perp_J{\cal{A}}^\perp_{J'}$ commutes with ${\cal{A}}^\perp_K$. If one of them commutes and the other anti-commutes, we get anti-commutativity. This implies that subgroups can be combined to generate other subgroups in an easy way.

\begin{lem} The subgroup structure of $S^{\perp}$ (and therefore also  of $S$) is independent of the direction of the arrows in $G$: those graphs, $G$, sharing the same $G_b$ will have the same commutative subgroup structure (although the precise group members differ).
\end{lem}
\begin{pf}
The subgroup structure is only dependent on $G_b$, which is the same regardless of arrow direction.
\end{pf}

%If any row of $\Gamma$ is linearly dependent on other rows, we can delete the row and the corresponding column, and work with the reduced adjacency matrix,
%$\tilde{\Gamma}$,  induced by the remaining nodes. The corresponding subgraph will have the same commutativity/anti-commutativity properties, and the number of elements in any given maximal commutative subset will be halved. \end{lem}
%\begin{pf}
%Let $r_k$ be row $k$ of $\Gamma$.
%Let $r_j=\bigoplus_{k\in V}r_k$. Then ${\cal N}_j={\cal N}_V$, %where the neighbourhood of node $j$ is equal to the neighbourhood of the %composite node $\prod_{D}i$, 
%so the commutativity/anti-commutativity of ${\cev{s_j}}$ is the same as that of ${\cev{s_V}}=\prod_{k\in V}{\cev{s_k}}$. %This means that any element %commutes with ${\cev{s}_{\prod_{D}i}}$ iff it commutes with ${\cev{s}_j}$.
%\end{pf}

\vspace{2mm}

From lemma \ref{comm}, any two vectors $v,v'\in \F_2^n$ that satisfy $v\Gamma v'^T = 0$ define two commuting members of $S^{\perp}$. From here it follows that, if $\Gamma$ has rank $n-t$, then
%, from lemma \ref{iso},
$S^\perp$ comprises $2^t$ copies of a group that is isomorphic to ${\tilde{S^\perp}}$, where ${\tilde{S^\perp}}$ is generated from the rows of ${\tilde{\cal A}}^\perp$, where ${\tilde{\cal A}}^\perp$ is obtained from ${\cal A}^\perp$ by deleting $t$ dependent rows and corresponding columns of ${\cal A}^\perp$. We then consider the maximum commuting subgroups of ${\tilde{S^\perp}}$ instead of $S^\perp$.
Equivalently we reduce $\Gamma$ to ${\tilde{\Gamma}}$, being a $n-t \times n-t$ matrix of maximum rank, and look for two vectors $v,v'\in \F_2^{n-t}$ that satisfy $v{\tilde{\Gamma}} v'^T = 0$.
More generally, by linearity, there exist size $2^h$ subgroups of $F_2^{n-t}$ (linear subspaces), $L_B$, as generated by $h \times n$ binary matrices, $B$, that satisfy 
$B{\tilde{\Gamma}} B^T = 0$. Each such linear subspace, $L_B$, defines a commuting subgroup of ${\tilde{S^\perp}}$.

\begin{lem} \label{iso}
Let $\Gamma$ have rank $n-t$, and
let $\tilde{\Gamma}$ be a $n - t \times n - t$ maximum rank matrix obtained from $\Gamma$ by removing $t$ linearly dependent rows and corresponding columns. There are
multiple choices for $\tilde{\Gamma}$. Then each commuting subgroup, ${\tilde{P}}$ of ${\tilde{S^\perp}}$, as described by some $B$ satisfying $B{\tilde{\Gamma}} B^T = 0$, is in one-to-one correspondence with a commuting subgroup, $P$, of $S^\perp$, as described by some ${\hat{B}}$ satisfying ${\hat{B}}\Gamma{\hat{B}}^T = 0$, where $|P| = 2^t|{\tilde{P}}|$.
\end{lem}
\begin{pf}
Let $Q$ be a set of $t$ linearly dependent rows of $\Gamma$.
For some $u,u' \in \F_2^n$ let $u\Gamma u'^T = 0$. Then $u,u' \in L_{\hat{B}}$ for some ${\hat{B}}$. Then $\exists w,w'$, uniquely, such that $u\Gamma = w\Gamma$ and $u'\Gamma = w'\Gamma$,
where $w$ and $w'$ are both zero at all positions defined by elements in $Q$. Then we can delete those $t$ rows of $\Gamma$, defined by the elements in $Q$, and the
corresponding columns, as $\Gamma$ is symmetric, and the corresponding elements of $w,w'$ to obtain and work with $\tilde{\Gamma}$ and $v,v'$, where $v,v' \in L_B$. Due to
the unique mapping from $u,u' \rightarrow w,w' \rightarrow v,v'$,  the commutative subgroups from $\tilde{\Gamma}$ are in one-to-one correspondence with those from
$\Gamma$, but $2^t$ times smaller due to the restriction to $0$ for $t$ positions, specified by $Q$, in $u,u'$.
\end{pf}

\vspace{2mm}

So $2^t$ copies of a commuting subgroup of ${\tilde{S^\perp}}$ comprise a commuting subgroup of  
$S^\perp$. We are particularly interested in the maximum size commuting subgroups of ${\tilde{S^\perp}}$ and $S^\perp$, i.e. in those subgroups that cannot be contained in larger commuting subgroups. In terms of the binary representation, this translates into finding matrices $B$ where $h$ is maximised, and such that $B{\tilde{\Gamma}} B^T = 0$.
Thus our question simplifies to finding the largest $h \times n-t$ $B$ matrices such that $B{\tilde{\Gamma}}B^T = 0$. Remember that $e = \frac{\m{rank}({\Gamma})}{2} = \frac{n-t}{2}$.

\begin{thm}
All maximum commuting subgroups of ${\tilde{S^\perp}}$ and $S^\perp$ are of size $2^e$ and $2^{n-e}$, respectively.
Moreover there are $\chi_e = \prod_{j=1}^e (2^j + 1)$ such groups in both cases, and each element is in $ \prod_{j=1}^{e-1} (2^j + 1)$ such groups.
\label{MaxCommSub}
\end{thm}
\begin{pf}
Let the rows, $B_j$, of $B$ generate the linear subspace, $L$.
Observe that ${\tilde{\Gamma}}$ has a zero diagonal. Therefore any $B_j \in F_2^{n-t}$ satisfies
$B_j{\tilde{\Gamma}}B_j^T = 0$, and therefore so does any member of $L$. The maximum size of $L$ is then derived as follows.
\bite
	\item Choose non-zero $B_0 \in F_2^{n-t}$ such that $B_0{\tilde{\Gamma}}B_0^T = 0$. There are $2^{n-t} - 1$ such choices.
	\item Choose non-zero $B_1 \in F_2^{n-t}$, $B_1 \ne B_0$, such that any member, $b$, of the linear space generated by $\{B_0,B_1\}$ satisfies
	$b{\tilde{\Gamma}}b^T = 0$. There are $2^{n-t} - 2^1$ such choices.
	\item Choose non-zero $B_2 \in F_2^{n-t}$, $B_2$ not in the linear space generated by $\{B_0,B_1\}$, such that any member, $b$, of the linear space generated by
	$\{B_0,B_1,B_2\}$ satisfies $b{\tilde{\Gamma}}b^T = 0$. There are $2^{n-t} - 2^2$ such choices.
	\item \ldots and so on.
\eite
We continue in this manner until there are no more choices. Thus $B$ always has a maximum of $e$ rows, i.e. a maximum commuting subgroup of ${\tilde{S^\perp}}$ is always size $2^e$. In total we generate $\prod_{j=0}^{e-1} (2^{n-t-j} - 2^j)$ linear subspaces.
But, for $M$ a maximum rank $e \times e$ binary matrix, $B' = MB$ generates the same linear subspace as $B$. We have to remove such repetitions. $M$ is chosen
from $\m{GL}(e,2)$ and $|\m{GL}(e,2)| = \prod_{j=0}^{e-1} (2^e - 2^j)$, so the total number of unique linear subspaces, i.e. of maximum commuting subgroups of ${\tilde{S^\perp}}$, is
$\chi_e = \frac{\prod_{j=0}^{e-1} (2^{n-t-j} - 2^j)}{\prod_{j=0}^{e-1} (2^e - 2^j)} = \prod_{j=1}^e (2^j + 1)$. For $S^\perp$ we simply multiply the maximum commutative subgroup size of  ${\tilde{S^\perp}}$
by $2^t$, as there are $t$ redundant rows/columns, to get $2^{n-e} = 2^e \times 2^t$. The number of maximum commutative subgroups remains at $\chi_e = \prod_{j=1}^e (2^j + 1)$.

Each element is in $\prod_{j=1}^{e-1} (2^{n-t-j} - 2^j)$ linear subspaces, since we fix the first element. There are $ \prod_{j=1}^{e-1} (2^e - 2^j)$ repetitions (the elements in GL$(e,2)$ with fixed first row), so each element is in $ \frac{\prod_{j=1}^{e-1} (2^{n-t-j} - 2^j)}{\prod_{j=1}^{e-1} (2^e - 2^j)} = \prod_{j=1}^{e-1} (2^j + 1)$ such groups.
\end{pf}

\vspace{2mm}

It is interesting to note that $\chi_e$  is also the total number of binary self-dual codes of length $2(e+1)$. This leads us to make a passing observation regarding the multiplicative order of ${\tilde{\Gamma}}$.

\begin{lem}
Let $u$ be the multiplicative order of ${\tilde{\Gamma}}$, i.e. let ${\tilde{\Gamma}}^u = I$ for some minimum positive $u$. Then $u$ is even and ${\tilde{\Gamma}}$ cannot
be factored as ${\tilde{\Gamma}} = \Omega \Omega^T$ for some $\Omega$.
\end{lem}
\begin{pf}
If $u$ is odd or if $\exists \Omega$ such that ${\tilde{\Gamma}} = \Omega \Omega^T$ then we show that $\exists B$ matrices of size $\frac{n-t}{2} \times n-t$ such that $B\Gamma B^T = 0$, where the set of $B$
matrices is in one-to-one correspondence with the set of self-dual binary codes of length $2e$ and dimension $e$, implying that there are $\chi_{e-1}$ of them. But this is impossible as, by theorem \ref{MaxCommSub}, there are $\chi_e$ of them. The argument is as follows.
If $u$ is odd then $B{\tilde{\Gamma}}B^T = C{\tilde{\Gamma}}^{\frac{u-1}{2}}{\tilde{\Gamma}}{\tilde{\Gamma}}^{\frac{u-1}{2}}C^T = CC^T = 0$, where $B =  C{\tilde{\Gamma}}^{\frac{u-1}{2}}$.
If ${\tilde{\Gamma}} = \Omega \Omega^T$ then $B{\tilde{\Gamma}}B^T = CC^T = 0$, where $C = B\Omega$. In both cases $C$ is taken from the set of matrices that generate self-dual
binary codes of length $n-t$ and dimension $\frac{n-t}{2}$, where $B = C{\tilde{\Gamma}}^{\frac{u-1}{2}}$ in the first case and $B = C\Omega^{-1}$ in the second case. There are
$\chi_{e-1}$ such matrices, which contradicts theorem \ref{MaxCommSub}.
%Thus, given these conditions on ${\tilde{\Gamma}}$, all maximum commuting subgroups of ${\tilde{S^\perp}}$ and $S^\perp$ are of size $2^{\frac{n-t}{2}} = 2^e$ and $2^{n-e}$, respectively.
\end{pf}

\begin{ex} Let $G$ be the graph defining a mixed graph state, $\rho$, stabilised by the rows of ${\cal{A}} = \left ( \begin{array}{c}
X\otimes Z\otimes Z\otimes Z \\
\,I\otimes X\otimes I\otimes I \\
\,I\otimes Z\otimes X\otimes Z, \\
\,I\otimes I\otimes Z\otimes X \end{array} \right )$.
We want to find all maximal commutative subgroups of $S^\perp$. The adjacency matrix of $G_b$ is $\Gamma=\left(\begin{array}{cccc}
0&1&1&1\\
1&0&1&0\\
1&1&0&0\\
1&0&0&0
\end{array}\right)$, which has full rank, so ${\tilde{\Gamma}} = \Gamma$. Therefore $t = 0$, $e = \frac{\m{rank}(\Gamma)}{2} = \frac{n-t}{2} = 2$, and
the size of a maximal commutative subgroup of $S^{\perp}$ is $2^{n-e}=2^2=4$. Moreover, $\Gamma^u = \Gamma^4 = I$.
%The binary linear code generated by this matrix is $C=\{0000,0111,1010,1100,1000,1101,1011,1111,0110,0010,0100,0001,0011,0101,1110,1001\}={\mathbb{F}}_2^4$ (this is obvious since $\Gamma_{G_b}$ has full rank). %, The order within the set will determine who commutes with whom, it is not the same if, say, 0001 is the element corresponding to ${\cev{s}_0}$, or to ${\cev{s}_2}{\cev{s}_3}$.

From theorem \ref{MaxCommSub} there are $\chi_e = 15$ maximum commuting subgroups of $S^{\perp}$ and these can be generated by 15 matrices, $B$, satisfying
$B{\tilde{\Gamma}}B^T = 0$, where the $B$ can be chosen, (non-uniquely) from
\begin{scriptsize}
$$ \begin{array}{ccccc}
\left ( \begin{array}{c} 1000 \\ 0110 \end{array} \right ), & \left ( \begin{array}{c} 1000 \\ 0101 \end{array} \right ), & \left ( \begin{array}{c} 1000 \\ 0011 \end{array} \right ), &
\left ( \begin{array}{c} 0100 \\ 1010 \end{array} \right ), & \left ( \begin{array}{c} 0100 \\ 0001 \end{array} \right ), \\
\left ( \begin{array}{c} 0100 \\ 1011 \end{array} \right ), & \left ( \begin{array}{c} 1100 \\ 0010 \end{array} \right ), & \left ( \begin{array}{c} 1100 \\ 1001 \end{array} \right ), &
\left ( \begin{array}{c} 0010 \\ 0001 \end{array} \right ), & \left ( \begin{array}{c} 0010 \\ 1101 \end{array} \right ), \\
\left ( \begin{array}{c} 1010 \\ 1001 \end{array} \right ), & \left ( \begin{array}{c} 1010 \\ 1101 \end{array} \right ), & \left ( \begin{array}{c} 0110 \\ 0001 \end{array} \right ), &
\left ( \begin{array}{c} 0110 \\ 1001 \end{array} \right ), & \left ( \begin{array}{c} 1011 \\ 0111 \end{array} \right ).
\end{array} $$
\end{scriptsize}

For instance, $B = \left ( \begin{array}{c} 1100 \\ 1001 \end{array} \right )$ acting multiplicatively on the rows of ${\cal{A}}^T$ generates the maximum commuting
subgroup with elements $\{I \otimes I \otimes I \otimes I, \mz \mp iY \otimes X \otimes Z \otimes I, \mz \mp i Y \otimes I \otimes Z \otimes X, \mz I \otimes X \otimes I \otimes X\}$.
Similarly, $B = \left ( \begin{array}{c} 1011 \\ 0111 \end{array} \right )$ acting multiplicatively on the rows of ${\cal{A}}^T$ generates the maximum commuting
subgroup with elements $\{I \otimes I \otimes I \otimes I, \mz \mp iX \otimes I \otimes Y \otimes Y, \mz \mp i Z \otimes X \otimes X \otimes Y, \mz \mp iY \otimes X \otimes Z \otimes I\}$.
Observe that, in this case, $\mp iY \otimes X \otimes Z \otimes I$ occurs in both subgroups. More generally,  each group element occurs in 3 maximal commutative subgroups.
\end{ex}

\begin{ex} Let $G$ be the graph defining a mixed graph state, $\rho$, stabilised by the rows of ${\cal{A}} = \left ( \begin{array}{c}
X\otimes Z\otimes I \\
\,I\otimes X\otimes Z \\
Z \otimes I \otimes X \end{array} \right )$.
The adjacency matrix of $G_b$ is $\Gamma=\left(\begin{array}{ccc}
0&1&1\\
1&0&1\\
1&1&0
\end{array}\right)$, which has rank 2. Therefore $t = 1$, and $e = \frac{n-t}{2} = 1$, and
the size of a maximal commutative subgroup of $S^{\perp}$ is $2^{n-e}=2^2=4$. By removing one dependent row and corresponding column
of $\Gamma$ (we choose the last row/column) we can obtain a ${\tilde{\Gamma}} =\left(\begin{array}{cc}
0&1\\
1&0
\end{array}\right)$.

From theorem \ref{MaxCommSub} there are $\chi_e = 3$ maximum commuting subgroups of $S^{\perp}$ and these can be generated by extensions of 3 matrices, $B$, satisfying $B{\tilde{\Gamma}}B^T = 0$, where the $B$ are
\begin{scriptsize}
$$ \begin{array}{ccc}
\left ( \begin{array}{c} 10 \end{array} \right ), & \left ( \begin{array}{c} 01 \end{array} \right ), & \left ( \begin{array}{c} 11 \end{array} \right ).
\end{array} $$
\end{scriptsize}

Noting that the last row/column of $\Gamma$ is dependent because $\Gamma_0 + \Gamma_1 + \Gamma_2 = 000$, we obtain the three maximum commuting
subgroups of $S^{\perp}$ via the 3 matrices, $B'$, these being
\begin{scriptsize}
$$ \begin{array}{ccc}
\left ( \begin{array}{c} 100 \\ 111 \end{array} \right ), & \left ( \begin{array}{c} 010 \\ 111 \end{array} \right ), & \left ( \begin{array}{c} 110 \\ 111 \end{array} \right ),
\end{array} $$
\end{scriptsize}

leading to the following three maximum commuting subgroups of $S^{\perp}$, respectively:
$\{I \otimes I \otimes I, X \otimes I \otimes Z, \mp iY \otimes IY\otimes Y, \mp iZ \otimes Y \otimes X\}$,
$\{I \otimes I \otimes I, Z \otimes X \otimes I, \mp iY \otimes IY\otimes Y, \pm iX \otimes Z \otimes Y\}$,
$\{I \otimes I \otimes I, \mp Y \otimes X \otimes Z, \mp iY \otimes IY\otimes Y, I \otimes Z \otimes X\}$.
\end{ex}

\begin{lem}
Let $G$ and $H$ be two mixed graphs on $n$ variables. If they have the same mixed rank $e$, there exists an isomorphism between  their commutative subgroups.
\end{lem}
\begin{pf} Assume first $n=2e$. Then, the linear group generated by their respective $\Gamma$ is in both cases ${\mathbb{F}}_2^e$. The isomorphism is given by sending the basis elements of $\left<\Gamma^G\right>$ (i.e., the rows of $\Gamma^G$) to the corresponding elements in $\left<\Gamma^H\right>$. Their commutative properties are the same, so this isomorphism preserves commutatitivity.

For $n>2e$, the linearly dependent elements can be expressed as $\Gamma^G_j=\sum_K \Gamma^G_k$, and similarly (after a possible reordering, that is, a permutation of the nodes) $\Gamma^H_j=\sum_{K'} \Gamma^H_k$. The isomorphism for the independent rows is the same as before. For the dependent rows, we send ${\cal{A}}^T_j{\cal{A}}^T_K$ to ${\cal{A}}^T_j{\cal{A}}^T_{K'}$. This sends the row $00\ldots0$ to itself, and commutativity is preserved by linearity. 
\end{pf}
\begin{ex}
Let $G_b$ be the line in 4 variables, with adjacency matrix $\Gamma =   \left(\begin{array}{cccc}
0&1&0&0\\
1&0&1&0\\
0&1&0&1\\
0&0&1&0
\end{array}\right)$,
and the complete graph in 4 variables. Both of them have $e=2$. The group generated by the respective adjacency matrices is $\mathbb{F}_2^4$. We
have the isomorphism $f:S^T_{\m{Line}}\rightarrow S^T_{\m{Complete graph}}$ defined by $f({\cal{A}}^T_0)={\cal{A}}^T_{\{0,2,3\}},\,f({\cal{A}}^T_1)={\cal{A}}^T_{\{0,2\}},\,f({\cal{A}}^T_2)={\cal{A}}^T_{\{1,3\}},\,f({\cal{A}}^T_3)={\cal{A}}^T_{\{0,1,3\}}$. 
%the elements that commute with ${{\cal{A}}^T_0}$ are  $\{I,{\cal{A}}^T_0},{{\cal{A}}^T_2},{{\cal{A}}^T_3},{{\cal{A}}^T_{02}},{{\cal{A}}^T_{03}},{{\cal{A}}^T_{23}},{{\cal{A}}^T_{023}}\}$, which are generated by ${{\cal{A}}^T_0},{{\cal{A}}^T_2},{{\cal{A}}^T_3}$, with corresponding matrix rows $0100,0101,0010$. For the clique these generators are ${{\cal{A}}^T_{023}},{{\cal{A}}^T_{13}},{{\cal{A}}^T_{013}}$, generating an isomorphic group.
In particular, the %commutativity of ${{\cal{A}}^T_0$ on the line is equivalent to the commutativity of ${{\cal{A}}^T_{023}$ on the clique, up to a reordering of the group elements.
commutative subgroup on the line generated by ${\cal{A}}^T_0,{\cal{A}}^T_3$ corresponds to the commutative subgroup on the complete graph generated by   ${\cal{A}}^T_{\{0,2,3\}},{\cal{A}}^T_{\{0,1,3\}}$.
\end{ex}
\begin{ex} let $G$ be the triangle on 3 variables, and $H$ the star graph on 3 variables. For both of them, $e=1$. A possible isomorphism is given by $f:S^T_G\rightarrow S^T_H$, defined by $f({\cal{A}}^T_0)={\cal{A}}^T_0,\,f({\cal{A}}^T_1)={\cal{A}}^T_{1},\,f({\cal{A}}^T_2)={\cal{A}}^T_{\{0,2\}}$. In particular, the  commutative subgroup on  $G$ generated by ${\cal{A}}^T_1,{\cal{A}}^T_{\{0,2\}}$ corresponds to the commutative subgroup on $H$ generated by   ${\cal{A}}^T_{1},{\cal{A}}^T_{2}$. Note that if we take the isomorphism given by $f:S^T_G\rightarrow S^T_H$, defined by $f({\cal{A}}^T_0)={\cal{A}}^T_0,\,f({\cal{A}}^T_1)={\cal{A}}^T_{2},\,f({\cal{A}}^T_2)={\cal{A}}^T_{\{0,1\}}$, we obtain the same commutative subgroups.

\end{ex}

\section{Children of pure graph state parents}\label{sectionchildren}
In this section, given a binary string of length $n$, $j=j_0\cdots j_{n-1}$, we define
%${\cev{s}_L}$ by%: ${\cev{s}_{0\ldots 1\ldots0}}={\cal A}_j^T$, where the 1 is in postion $j$. In general, 
${\cev{s}_j}:={\cal A}_K^T$, where $K=\{i: j_i=1\}$, where ${\cev{s}_{0\ldots 00}} = I$. Remember that, for $i,i' \in K$, $i' > i$, we define ${\cal A}_K^T$ to contain the multiplicative
factor ${\cal A}^T_{i'}{\cal A}^T_i$. However the ordering ${\cal A}^T_i{\cal A}^T_{i'}$ is just as valid and differs by a global multiplicative factor of $-1$ if ${\cal A}^T_i$ and
${\cal A}^T_{i'}$ anti-commute. This situation carries over to the ${\cev{s}_j}$. In theorem \ref{pauli sum}, and for such cases where ${\cev{s}_j}$ contains anti-commuting factors, the $b_j$ are therefore only defined up to $\pm 1$.
%But in theorem \ref{pauli sum} this $\pm 1$ factor is subsumed into the $b_j$ so can ignored in this context.
Let $L_m=\{v : ({\cal{A}}^e)_{v,n+m}=Z\mbox{ or } Y\}$, and let
$J = \{x \in {\mathbb{F}}_2^n \mz | \mz {\cal I}(x) = 1\}$, where ${\cal I}(x) = \prod_{m=0}^{e-1}(1 + \sum_{k\in L_m} x_k) = \prod_{m=0}^{e-1} {\cal I}_m(x)$.

\begin{thm} \label{pauli sum} Given a child, $\rho$,  of a pure graph state parent of a mixed graph, corresponding to a symmetric extension of ${\cal{A}}$ by $e$ columns and rows, ${\cal{A}}^e$, we can write $\rho$ as:
%$$\rho=\sum_{u\in J} \pm \sigma_u$$ where $J$ is the indicator of the Boolean function $\displaystyle\prod_{j=0}^{t-1}(\sum_{k\in L_j} x_k+1)$,  where $L_j=\{v : ({\cal{A}}^e)_{v,n+j}=Z\mbox{ or } Y\}$,   and $\sigma_u$ is the $u$th element in $S^\perp$.

$$\rho= \frac{1}{2^n}\sum_{j\in J}b_j{\cev{s}_j}, \mf \mf b_j \in \{\pm 1, \pm i\}, \forall j, $$%,  \mf
%{\cev{s}_j} \in S^{\perp}, $$
where
%$J$ is the indicator\begin{footnote}{By the indicator of a Boolean function $f$, we mean the set of binary linear strings, $L$, such that $f(L)=1$.%That is, the set of binary linear strings $L\in{\mathbb{F}}_2^n$ where the truth table of the function is 1.
%}\end{footnote}of the Boolean function $\displaystyle\prod_{m=0}^{e-1}(\sum_{k\in L_m} x_k+1)$,  where
%$L_m=\{v : ({\cal{A}}^e)_{v,n+m}=Z\mbox{ or } Y\}$,  %for  a symmetric extension of ${\cal{A}}$ by $e$ columns, 
%$b_j=\pm\frac{1}{2^n},\pm\frac{i}{2^n}$, and
%Furthermore, $\pm,\pm i \{\sigma_u: u\in J\}$  is a commutative subgroup of $S^\perp$, and the sign of summation is given so that if $\sigma_u,\sigma_v$ are present in the sum, so is $\sigma_u\sigma_v$ (and not $-\sigma_u\sigma_v$).
all ${\cev{s}_j},\,j\in J$, commute pairwise, and if ${\cev{s}_j},{\cev{s}_{j'}}$ are present in the sum, so is ${\cev{s}_j}{\cev{s}_{j'}} = {\cev{s}_{j+j'}}$, implying that $\{{\cev{s}_j}: j\in J\}$ 
is a commutative subgroup of $S^\perp$.
\end{thm}

Let $j^{\al}  \in {\mathbb{F}}_2^n$ and $j^{\al,\beta}  \in {\mathbb{F}}_2^n$, $0 \le \al,\beta < n$ be weight-one and weight-two binary vectors with 1's only in positions $\al$, and $\al$ and $\beta$, respectively.

%THIS I CHANGED THE 23.7
{\begin{lem}\label{sign} Given a  child, $\rho$, of a pure graph state parent of a mixed graph, described by ${\cal A}^e$, the coefficients, $b_j$, in the sum of theorem \ref{pauli sum}, are as follows:

\begin{itemize} \item Case $e=1$:

%We decompose the terms in the sum to the ones with
\begin{itemize}
%\item Any term, $\prod_{u\in A}\sigma_u, \,A\subseteq \{0,\ldots,n-1\}$, with commuting $\sigma_u$, will have  coefficient +1.
%\item Any term, $\prod_{u\in A}\sigma_u, \,A\subseteq \{0,\ldots,n-1\}$, where some of the $\sigma_u$ do not commute, will have  coefficient $\pm i$, and the sign will be given by the order of the product $\prod_{u\in A\cap B}\sigma_u\prod_{u\in A\cap C}\sigma_u$, where $A=\{u:({\cal{A}}^e)_{u,n}=Z$
\item If ${\cev{s}_{j^{\al}}}$
% the matrix ${\cal{A}}_i^T, \,0\leq i < n$,
is present in the sum % (that is, if $b_j=\pm\frac{1}{2^n}$), 
then $b_{j_{\al}} = +1$.
\item If ${\cev{s}_{j^{\al,\beta}}}$
%the matrix ${\cal{A}}^T_{\{j,k\}}\,0\leq j,k < n$,
is present in the sum then we have two cases:
\begin{enumerate}
\item if ${\cev{s}_{j^{\al}}}$ and ${\cev{s}_{j^{\beta}}}$ anticommute, then $\exists $ $a,b\in\{\al,\beta\}$ such that $({\cal{A}}^e)_{a,n}=Z$ and  $({\cal{A}}^e)_{b,n}=Y$; furthermore, the corresponding term in the sum will be equal to $\pm i{\cal{A}}^T_{\{\al,\beta\}}=i{\cal{A}}_a^T{\cal{A}}_b^T$, so $b_{j^{\al,\beta}} = \pm i$.%, where the sign  is given by the order of multiplication ${\cal{A}}_a{\cal{A}}_b$.%, by which we mean that we multiply on the left the matrix $\sigma_a$ that has $({\cal{A}}^e)_{a,n}=Z$, and on the right the matrix $\sigma_b$ that has $({\cal{A}}^e)_{b,n}=Y$.
\item if ${\cal{A}}_{\al}^T$ and ${\cal{A}}_{\beta}^T$  commute,   the corresponding term in the sum will be equal to ${\cal{A}}^T_{\{\al,\beta\}}$, and $b_{j^{\al,\beta}} = 1$.%, then its coefficient is always $\frac{1}{2^n}$.
\end{enumerate}
\item Any other matrices, ${\cal{A}}_K^T$, will be products of ${\cal{A}}_{K_j}^T$ also present in the sum, with $K_j$ of size 1 and size 2. Furthermore, their coefficients are given by the multiplication of the coefficients of the terms into which it is decomposed.
%\item For a given child, the sign of the term, $\prod_u\in A$, present in $\rho$, is equal to +1 in all commuting cas

\end{itemize}
\item General $e$: The coefficient of ${\cal{A}}_{K}^T$, with $K$ of size 1 or 2 is given by the multiplication together of the coefficients, one from each extension column. In this case, there might be nondecomposable terms. However, the coefficient of a term will be giving by multiplying the coefficients obtained by taking each column separately.
\item Furthermore, for any $e$,  if we change the sign of a term ${\cal{A}}_k^T$, $0\leq k < n$,  then we add the Boolean linear term $x_k$ to the quadratic Boolean representation of the pure parent graph state, or if we change the sign of ${\cal{A}}_K^T$, for $K$ a set of size $t>1$ present in the sum, we add any of the Boolean linear terms $x_k$ for each $k$ s.t. $({\cal{A}}^e)_{n+j,k}=Y$ or $({\cal{A}}^e)_{n+j,k}=Z$.%; then, the signs of all $\sigma_v$  change sign if supp$(j)\cap\mbox{supp}(v)\neq\emptyset$.
\end{itemize}
\end{lem}

\begin{ex} Let the directed triangle be defined by the stabilizer basis ${\cal{A}}_0=X\otimes Z\otimes I,\,{\cal{A}}_1=I\otimes X\otimes Z,\,{\cal{A}}_2=Z\otimes I\otimes X$. Then the basis of $S^\perp$ is given by reversing the arrows: ${\cal{A}}_0^T=X\otimes I\otimes Z,\,{\cal{A}}_1^T=Z\otimes X\otimes I,\,{\cal{A}}_2^T=I\otimes Z\otimes X$, so
$S^\perp=\{\cev{s}_{000}=I\otimes I\otimes I,\cev{s}_{100}=X\otimes I\otimes Z,\,\,\cev{s}_{010}=Z\otimes X\otimes I,\,\cev{s}_{110}=-i Y\otimes X\otimes Z,
\,\cev{s}_{001}=I\otimes Z\otimes X,\cev{s}_{101}=i X\otimes Z\otimes Y,\,\cev{s}_{011}=-i Z\otimes Y \otimes X,\,\cev{s}_{111}=-i Y \otimes Y \otimes Y\}$.

One of the two parents is formed by adding the column $(X\ Z\ Y)^T$ to $\cal{A}$, giving the parent \\$i^{2(x_0x_2+x_1x_2+x_1x_3+x_2x_3+x_2)+x_2} = i^{2(x_0x_2+x_1x_2+x_2+({\cal I}(x)+1)x_3)+x_2}$. Here $L_3=\{1,2\}$, and ${\cal I}(x) = x_1+x_2+1$, so $J= < 100, 011> = \{000,100,011,111\}$. We can re-interpret ${\cal I}(x)$ and $J$ as parity and generator matrices, $H$ and
$G$, respectively, where $H = \left ( \begin{array}{c} 011 \end{array} \right )$ and $G = \left ( \begin{array}{c} 100 \\ 011 \end{array} \right )$, and where $G$ generates the binary linear code with
codewords in set $J$.
By tracing over the environmental qubit, $x_3$, we get $$\rho=\frac{1}{8}\left(I\otimes I\otimes I+X\otimes I\otimes Z+Z\otimes Y\otimes X+Y\otimes Y\otimes Y\right)=\frac{1}{8}\left(\cev{s}_{000}+\cev{s}_{100}+i\cev{s}_{011}+i\cev{s}_{111}\right).$$
Note that $X\otimes I\otimes Z$ and $Z\otimes Y\otimes X$ commute, and that $(X\otimes I\otimes Z)\cdot(Z\otimes Y\otimes X)=Y\otimes Y\otimes Y$.
We can obtain another parent for the same child by doing a local complementation (LC) on vertex 3 (i.e. LC$_3$) on the graph described by the parent. This is an LC in the environment. If one describes the $\Z_4$-linear offsets of the parent by black vertices in the graph, then LC$_3$ swaps the vertex neighbours of vertex 3 from black $\leftrightarrow$ white. 

\vspace{2mm}

The other of the two parents is formed by adding the column $(X\ Y\ Z)^T$ to $\cal{A}$, giving the parent $i^{2(x_0x_2+x_1+x_1x_3+x_2x_3)+x_1} = i^{2(x_0x_2+x_1+({\cal I}(x)+1)x_3)+x_1}$. Once again $L_3=\{1,2\}$, ${\cal I}(x) = x_1+x_2+1$, and $J=\{000,100,011,111\}$. By tracing over the environmental qubit, $x_3$, we get $$\rho=\frac{1}{8}\left(I\otimes I\otimes I+X\otimes I\otimes Z-Z\otimes Y\otimes X-Y\otimes Y\otimes Y\right)=\frac{1}{8}\left(\cev{s}_{000}+\cev{s}_{100}-i\cev{s}_{011}-i\cev{s}_{111}\right).$$ $H$ and $G$ are unchanged.
Observe that the second parent is obtained from the first by swapping the positions of $Y$ and $Z$ in rows 1 and 2 of the extra column (column 3). In terms of the Boolean function representation of the parents, this translates to adding the Boolean quadratic term $x_1x_2+x_1+x_2$, and removing $x_2$ and adding $x_1$ $\Z_4$-linear terms. %This is equivalent, graphically, to doing a local complementation (LC) on vertex 3 (i.e. LC$_3$) on the graph described by the parent. This is an LC in the environment. If one describes the $\Z_4$-linear offsets of the parent by black vertices in the graph, then LC$_3$ swaps the vertex neighbours of vertex 3 from black $\leftrightarrow$ white. 
Observe that the coefficients of the Pauli basis terms $Z\otimes Y\otimes X$ and
$Y\otimes Y\otimes Y$ are multiplied by $-1$. This is because both terms are generated from rows 1 and 2 of ${\cal{A}}^T$ which are the rows where $Y$ and $Z$ are swapped in ${\cal{A}}^e$.
\vspace{2mm}

For each parent we can also consider how the addition of binary linear terms (in the lab) affects the resultant child density matrix. We don't currently consider the addition of binary linear terms in the environment (i.e. $x_3$ for this example). For instance, for the addition of column $(X\ Z\ Y)^T$ then
$i^{2(x_0x_2+x_1x_2+x_1x_3+x_2x_3 + x_0+x_2)+x_2}$, i.e. the addition of $x_0$, simply flips the signs of $X\otimes I\otimes Z$ and $Y\otimes Y\otimes Y$ as both terms have row $0$
of ${\cal{A}}^T$ as a factor. However addition of $x_1$ or $x_2$ has the same effect as swapping $Y$ and $Z$ in rows 1 and 2, so swaps between the two parents. Thus the addition of ${\bar{\cal I}}(x) = {\cal I}(x) + 1$ fixes the child density matrix. So we have the following maps for $a \in \{0,1\}$:

\begin{small}
$$ \begin{array}{l|l}
\m{child } = 8\rho & \m{parent} \\ \hline
I\otimes I\otimes I+X\otimes I\otimes Z+Z\otimes Y\otimes X+Y\otimes Y\otimes Y & 2(x_0x_2+x_1x_2+x_2+{\bar{\cal I}}(x)(x_3 + a))+x_2, \\
   & 2(x_0x_2 +x_1+ x_2 + {\bar{\cal I}}(x)(x_3 + a)) + x_1, \\
I\otimes I\otimes I-X\otimes I\otimes Z+Z\otimes Y\otimes X-Y\otimes Y\otimes Y & 2(x_0x_2+x_1x_2+x_0+x_2+{\bar{\cal I}}(x)(x_3 + a))+x_2, \\
   & 2(x_0x_2 + x_0 +x_1+ x_2 + {\bar{\cal I}}(x)(x_3 + a)) + x_1, \\
I\otimes I\otimes I+X\otimes I\otimes Z-Z\otimes Y\otimes X-Y\otimes Y\otimes Y & 2(x_0x_2+x_1 + {\bar{\cal I}}(x)(x_3 + a))+x_1, \\
   & 2(x_0x_2 +x_1x_2+x_1+x_2+ {\bar{\cal I}}(x)(x_3 + a)) + x_2, \\
I\otimes I\otimes I-X\otimes I\otimes Z-Z\otimes Y\otimes X+Y\otimes Y\otimes Y & 2(x_0x_2+x_0+x_1 + {\bar{\cal I}}(x)(x_3 + a))+x_1, \\
   & 2(x_0x_2 + x_1x_2+x_0 +x_1+ x_2+{\bar{\cal I}}(x)(x_3 + a)) + x_2.
\end{array} $$
\end{small}

%Note also that
%$$\rho=\frac{1}{8}\left(I\otimes I\otimes I+X\otimes I\otimes Z-Z\otimes Y\otimes X-Y\otimes Y\otimes Y\right)$$ is a density matrix corresponding to the parent $i^{2(x_0x_2+x_1x_2+x_1x_3+x_2x_3+x_1)+x_2}$, $$\rho=\frac{1}{8}\left(I\otimes I\otimes I-X\otimes I\otimes Z+Z\otimes Y\otimes X-Y\otimes Y\otimes Y\right)$$ is a density matrix corresponding to the parent $i^{2(x_0x_2+x_1x_2+x_1x_3+x_2x_3+x_0)+x_2}$, and $$\rho=\frac{1}{8}\left(I\otimes I\otimes I-X\otimes I\otimes Z-Z\otimes Y\otimes X+Y\otimes Y\otimes Y\right)$$ is a density matrix corresponding to the parent $i^{2(x_0x_2+x_1x_2+x_1x_3+x_2x_3+x_1+x_0)+x_2}$,while for instance $$\rho=\frac{1}{8}\left(I\otimes I\otimes I+X\otimes I\otimes Z-Z\otimes Y\otimes X+Y\otimes Y\otimes Y\right)$$ is not a density matrix.

Other children given by the extensions $(Z\ X\ Y)^T$ and $(Z\ Y\ X)^T$  are respectively
 $$\begin{array}{c}
          %\rho_0=I\otimes I\otimes I+X\otimes I\otimes Z+Z\otimes Y\otimes X+Y\otimes Y\otimes Y\\
          \rho_1=\frac{1}{8}\left(I\otimes I\otimes I+aZ\otimes X\otimes I+bX\otimes Z\otimes Y+cY\otimes Y\otimes Y\right)=\frac{1}{8}\left(\cev{s}_{000}+a\cev{s}_{010}-bi\cev{s}_{101}+ci\cev{s}_{111}\right)\\[3pt]
          \rho_2=\frac{1}{8}\left(I\otimes I\otimes I+aY\otimes X\otimes Z+bI\otimes Z\otimes X+cY\otimes Y\otimes Y\right)=\frac{1}{8}\left(\cev{s}_{000}+ai\cev{s}_{110}+b\cev{s}_{001}+ci\cev{s}_{111}\right) \\
          \hspace{20mm} \m{ with condition } c = ab, \hspace{5mm} a,b \in \{1,-1\} \hspace{5mm} \m{ in both cases}. \\
         \end{array}$$
\end{ex}

\begin{ex} Let a mixed 6-clique graph be defined by the stabilizer basis
${\cal{A}}= \left ( \begin{array}{cccccc}
X &  Z &  Z &  Z &  Z &  Z \\
I &  X &  Z &  Z &  Z &  Z \\
I &  I &  X &  Z &  Z &  Z \\
I &  I &  I &  X &  Z &  Z \\
I &  I &  I &  I &  X &  Z \\
I &  I &  I &  I &  I &  X.
\end{array} \right )$.
Then the basis of $S^\perp$ is obtained from ${\cal{A}}^T$. The parents are obtained by adding $e = 3$ columns to ${\cal{A}}$, with subsequent addition of $e = 3$ rows. For instance, one parent is given by
$$ {\cal{A}}^e= 
\begin{small} \left ( \begin{array}{ccccccccc}
X &  Z &  Z &  Z &  Z &  Z &  X &  I &  X \\
I &  Y &  Z &  Z &  Z &  Z &  Y &  I &  X \\
I &  I &  X &  Z &  Z &  Z &  Z &  I &  X \\
I &  I &  I &  X &  Z &  Z &  I &  X &  Y \\
I &  I &  I &  I &  X &  Z &  I &  Y &  Y \\
I &  I &  I &  I &  I &  X &  I &  Z &  Y \\
I &  Z &  Z &  I &  I &  I &  X &  I &  I \\
I &  I &  I &  I &  Z &  Z &  I &  X &  I \\
I &  I &  I &  Z &  Z &  Z &  I &  I &  X
\end{array} \right ) \end{small} \equiv
\begin{small} \left ( \begin{array}{ccccccccc}
X &  I &  I &  I &  I &  I &  I &  I &  I \\
I &  Y &  I &  I &  I &  I &  Z &  I &  I \\
I &  I &  X &  I &  I &  I &  Z &  I &  I \\
I &  I &  I &  Y &  Z &  Z &  I &  I &  Z \\
I &  I &  I &  Z &  X &  Z &  I &  Z &  Z \\
I &  I &  I &  Z &  Z &  Y &  I &  Z &  Z \\
I &  Z &  Z &  I &  I &  I &  X &  I &  I \\
I &  I &  I &  I &  Z &  Z &  I &  X &  I \\
I &  I &  I &  Z &  Z &  Z &  I &  I &  X.
\end{array} \right ) \end{small}, $$
which represents $i^{2(x_1x_6 + x_2x_6 + x_3x_4 + x_3x_5 + x_3x_8 + x_4x_5 + x_4x_7 + x_4x_8 + x_5x_7 + x_5x_8+x_1+x_3+x_4+x_5) + x_1 + x_3 + x_5}$.
Here $L_3=\{1,2\}$, $L_4=\{4,5\}$, and $L_5=\{3,4,5\}$,
and ${\cal I}(x) = (x_1+x_2+1)(x_4 + x_5 + 1)(x_3 + x_4 + x_5 + 1)$, so $J= < 100000, 011000, 000011 > = \{000000, 100000, 011000, 111000, 000011, 100011, 011011, 111011 \}$.

We can re-interpret ${\cal I}(x)$ and $J$ as parity and generator matrices, $H$ and
$G$, respectively, where $H = \left ( \begin{array}{c} 011000 \\ 000011 \\ 000111 \end{array} \right )$ and
$G = \left ( \begin{array}{c} 100000 \\ 011000 \\ 000011 \end{array} \right )$, and where $G$ generates the binary linear code with
codewords in set $J$.
By tracing over the environmental qubits, $x_6, x_7, x_8$, we get
$$ \begin{array}{ll} \rho  = & (I\otimes I\otimes I \otimes I \otimes I \otimes I+X\otimes Z\otimes Z \otimes Z \otimes Z \otimes Z - I\otimes X\otimes Y \otimes I \otimes I \otimes I \\
  & - X\otimes Y\otimes X \otimes Z \otimes Z \otimes Z - I\otimes I\otimes I \otimes I \otimes X \otimes Y - X\otimes Z\otimes Z \otimes Z \otimes Y \otimes X \\
  & + I\otimes X\otimes Y \otimes I \otimes X \otimes Y + X\otimes Y\otimes X \otimes Z \otimes Y \otimes X). \end{array} $$

\end{ex}
%
%%I think a  commutative subgroup of Pauli can always be embedded in at least a stabilizer for a mixed graph, which would prove the first part of the conjecture.....

\begin{pf}(Theorem \ref{pauli sum}): The qubits present in each $L_m$ will be the ones connected with the environmental qubit $n+m$, implying that the difference between the measurement in $x_{n+m}=0$ and $x_{n+m}=1$ will be the Boolean function $\sum_{k\in L_m} x_k$. The support of $\rho$ will be equal to the binary vectors such that these Boolean functions are 0 for all $m$. %positions $w$ (for convenience, we take $w$ in natural numbers as opposed to the binary expansion), where $\rho_{0,w}=0$, and $\rho_{0,w'}=1$ for any other $w'$.
By lemma \ref{revmix}, $\rho$ can be expressed as the sum of some of the matrices on $S^\perp$; by lemma \ref{supp}, all the matrices in $S^\perp$ have nonintersecting support: $\rho$ will be equal to the sum of the matrices in $S^\perp$ whose support intersects the support of $\rho$; therefore, it is sufficient to find the nonzero entries of the first row of the matrix. It is easy to check that any matrix in $S^\perp$ has only one nonzero entry on the first row. In other words, $\rho$ will be equal to the sum of the matrices %that have  $\rho_{0,w'}=1$, that is, the sum of the matrices
whose  support intersects $\{(0,k), k\in K\}$, where $K=\{\sum_{m=0}^{n-1}2^mb_m, \forall b=b_0b_1\ldots  b_{n-1}\in J\}$, and  $J$ %(taken in natural numbers) 
is the indicator of the Boolean function $\displaystyle\prod_{m=0}^{e-1}(\sum_{k\in L_m} x_k+1)$. To prove that these matrices are the ${\cev{s}_j}$, where $j\in J$, we need the following lemma:

\begin{lem} Let ${\cev{s}_j}\in S^\perp$. Then, $({\cev{s}_j})_{0,t}=1\Leftrightarrow k=\sum_{m=0}^{n-1}2^mj_m$, for $j=j_0j_1\ldots j_{n-1}$. %where $u\in J$, 

\end{lem}

\begin{pf} Let ${\cev{s}_j}=\bigotimes_{R_X}X\bigotimes_{R_Y}Y\bigotimes_{R_Z}Z\bigotimes_{R_I}I$.  
Let $M$ be a $t\times t$ matrix. Then, $$M\otimes I=\left(\begin{array}{cc}
M&0\\
0&M\end{array}\right)\mbox{ and }M\otimes X=\left(\begin{array}{cc}
0&M\\
M&0\end{array}\right)$$
Note that $X$ has the same support as  $Y$, and $Z$ has the same support as $I$, so we only have to consider this two cases. This implies that support$(M\otimes I)=$support$(M)$, while $(0,k+2^{t-1})\in$support$(M\otimes X)\Leftrightarrow (0,k)\in$support$(M)$. This implies that $(0,k)\in$support$({\cev{s}_j})\Leftrightarrow k=\sum_{m\in R_X\cup R_Y}^{n-1}2^mm=\sum_{m=0}^{n-1}2^mj_m$, for $j=j_0j_1\ldots j_{n-1}$.

%Since $X$ has the same support as  $Y$, and $Z$ has the same support as $I$, ${\cev{s}_j}$ has the same support as ${\cev{s}_j}'=\bigotimes_{R_X\cup R_Y}X\bigotimes_{R_I\cup R_Z}I$. Note that supp$(j)=R_X\cup R_Y$, since any position in ${\cev{s}_j}\in S^\perp, j\neq0\ldots0$, is given by products of one $X$ and several $Z,I$, and the only matrix with just $Z$ and $I$ in $S^\perp$ is the identity, as the only possibility for obtaning $Z$ or $I$ is that no $X$ is present in the product. We shall prove the lemma  by induction on the size, $\alpha$, of $R_X\cup R_Y$:
%
%$\alpha=0$: Then ${\cev{s}_j}'={\cev{s}_j}=I_n$.  But that implies that  $j=0\ldots0$, and $({\cev{s}_{0\ldots0}})_{0,0}=1$.%that has support in $0,0$.
%
%Suppose it is true for $\alpha$, let us prove it is true for $\alpha+1$:
%
%$$\displaystyle{\cev{s}_j}'=\bigotimes_{R_X\cup R_Y}X\bigotimes_{R_I\cup R_Z}I=\left(\bigotimes_{(R_X\cup R_Y)\setminus\{v\}}X\bigotimes_{R_I\cup R_Z}I\right)\cdot X_v$$
%%Note ${\cev{s}_j}'$ corresponds to a ${\cev{s}_j}$ product of 
%$X_v$ has support in $(0,w)$, where $w=2^v$ (that is, where $x_v=1$ and $x_k=0\ \forall k\neq v$), and $\bigotimes_{(R_X\cup R_Y)\setminus\{v\}}X\bigotimes_{R_I\cup R_Z}I$ has support in $(0,w)$, where $w=\sum_{k\in (R_X\cup R_Y)\setminus\{v\}}2^k$.% (that is, where $x_k=1\ \forall k\in  (R_X\cup R_Y)\setminus\{v\}$, and $x_k=0$ elsewhere. 
%The product will have support in $(0,w)$, where $R_X\cup R_Y$, where $w=\sum_{k\in R_X\cup R_Y}2^k$.% $x_k=1\ \forall k\in  R_X\cup R_Y$, and $x_k=0$ elsewhere.
\end{pf}

To complete the proof of theorem \ref{pauli sum}, there remains to prove that all ${\cev{s}_j},\,j\in J$ commute pairwise, and if ${\cev{s}_j},{\cev{s}_k}$ are present in the sum, so is ${\cev{s}_j}{\cev{s}_k}$:%$\pm,\pm i \{\sigma_j: j\in J\}$  is a commutative subgroup of the Pauli group, and that if we change the sign in $\sigma_j$ from + to -, all the $\sigma_k$ change sign if supp$(j)\cap\mbox{supp}(k)\neq\emptyset$: First, we shall prove that $\pm,\pm i \{\sigma_j: j\in J\}$  is a commutative subgroup:

First, we shall prove that if ${\cev{s}_j},{\cev{s}_k}$ are present in the sum, so is ${\cev{s}_j}{\cev{s}_k}$: this follows from the indicator $J$ being a linear space, since ${\cev{s}_j}{\cev{s}_k}$ corresponds to the sum of the Boolean vectors.

%$J$,the indicator of $\displaystyle\prod_{m=0}^{e-1}(\sum_{k\in L_m} x_k+1)$, is equal to the intersection of the indicators of each of the terms in the product: $J=\cap_{m=0}^{e-1} J_m$, where $J_m$ is the indicator of $\sum_{k\in L_m} x_k+1$, since the product is 1 iff all the terms are equal to 1.  Now, each $J_m=\{j: \mbox{ supp}(j)\cap \mbox{supp}(L_m) \mbox{ is even}\}$. %To see that this implies that the sign of $\sigma_u\sigma_v$ is positive, note that by lemma \ref{sign}, changing the sign would imply adding a Boolean linear term, which couldn't be consistent with the sign of $\sigma_u$ or $\sigma_v$.
%
%Let ${\cev{s}_j},\,{\cev{s}_k}$ such that $j,k\in J$. Then, ${\cev{s}_j}{\cev{s}_k}=\pm{\cev{s}_{j+k}}$, and $j+k\in J$, since $j+k\in J_m\ \forall m$, since if $\mbox{ supp}(j)\cap \mbox{supp}(L_m) \mbox{ is even}$ and $\mbox{ supp}(k)\cap \mbox{supp}(L_m) \mbox{ is even}$,  $\mbox{ supp}(j+k)\cap \mbox{supp}(L_m) \mbox{ is even}\ \forall m$.
%Let ${\cev{s}_j}$ such that $j\in J$. Then, ${\cev{s}_j}^{-1}={\cev{s}_j}$, so the inverse is also in the set. These two together prove that $\pm,\pm i \{{\cev{s}_j}: u\in J\}$  is a  subgroup of the Pauli group\begin{footnote}{We include $\pm i$ so that this is a subgroup of the Pauli group, and not a subgroup of the group of tensor products of $X,Z,(XZ),I$}\end{footnote}.

Now we shall see that  it is commutative:
%$\sigma_j\sigma_k=\sigma_{j+k}=\sigma_{k+j}=\sigma_k\sigma_j$.
Let $j,k\in J$: %We can write $u=\sum a_l e_l$, where $a_l\in\{0,1\}$ and $e_l$ denotes the vector with all entries equal to 0 except for entry $l$, which is equal to 1. Then, an even number of the entries on each extension column are either $Y$ or $Z$
We can write ${\cev{s}_j}=\prod_{a\in A} {\cal{A}}_a^T$, and ${\cev{s}_k}=\prod_{b\in B} {\cal{A}}_b^T$, where ${\cal{A}}_a^T,{\cal{A}}_b^T$ are in the basis of $S^\perp$. Then, the corresponding ${\cal{A}}_a\in S, a\in A $, have an even number of $Y$ and $Z$ in each extension column, and similarly for $B$. Therefore, both ${\cev{s}_j}$ and ${\cev{s}_k}$ have either $X$ or $I$ in all extension columns, and therefore they commute.

\end{pf}

\begin{pf}[Lemma \ref{sign}]
Since all the parent Boolean functions have no constant terms and no linear terms involving the environmental qubits, the first entry in a truth table for each measurement $\ket{\psi_m}$ will always be +1. The first row of the density matrix for each measurement $\ket{\psi_m}$ will therefore be equal to $\frac{1}{2^{n/2}}\bra{\psi_m}$. Also, note that the final $\rho$ will be nonzero only where the entries for all $\ket{\psi_m}$ are equal, and will then be equal to any of them; it is therefore enough to look at $\bra{\psi_{0\ldots0}}$.
%\begin{itemize}
%\vspace{1cm}

{\em Case $e=1$:}
\begin{itemize}
\item Size 1: By the proof of theorem \ref{pauli sum}, any term, $\pm\frac{1}{2^n}{\cal{A}}^T_j,\,0\leq j\leq n-1$ is such that $({\cal{A}}^e)_{j,n}=X$ or $I$. %As we are extending to parent graphs without loops, that is to say, such that the eigenvalue is +1 for all the stabilizing group
Since row $j$ of the stabilizer of the parent graph state, ${\cal{A}}^e$, is equal to ${\cal{A}}\otimes$(extension entries), there is no linear term $x_j$ in the parent graph, so that the entry $0\ldots 1\ldots 0$, where the 1 occurs in position $u$, of $\ket{\psi_{0\ldots0}}$, and therefore of $\bra{\psi_{0\ldots0}}$ of the parent is $+\frac{1}{2^{n/2}}$. As the first row of ${\cal{A}}_j^T$ has $+\frac{1}{2^{n/2}}$ in the same position, the coefficient of ${\cev{s}_j}$ is $+\frac{1}{2^{n}}$.
\item Size 2:  Suppose that the matrix ${\cal{A}}^T_{\{j,k\}}\,0\leq j,k\leq n-1$, is present in the sum. Then:\begin{itemize}
\item  if ${\cal{A}}_j^T$ and ${\cal{A}}_k^T$  anti-commute, the term in the sum will be $\pm \frac{i}{2^{n}}{\cal{A}}^T_{\{j,k\}}\,0\leq j,k\leq n-1$,  since otherwise $\pm \frac{i}{2^{n}}{\cal{A}}_j^T{\cal{A}}_k^T\,0\leq j,k\leq n-1$ would not be Hermitian. Furthermore, neither ${\cal{A}}_j^T$ or ${\cal{A}}_j^T$ are present in the sum, because by theorem \ref{pauli sum} this would imply that both would be present and would therefore be commuting. This implies that  $\exists $ $a,b\in\{j,k\}$ such that $({\cal{A}}^e)_{a,n}=Y$ and  $({\cal{A}}^e)_{b,n}=Z$. %To investigate the sign  of $\pm \frac{i}{2^{n}}{\cal{A}}_{\{j,k\}}^T$, we only have to look at the pair  ${\cal{A}}_j,{\cal{A}}_k\in S$, since t
The first row of $\pm \frac{i}{2^{n}}{\cal{A}}_j^T{\cal{A}}_k^T$ has its only nonzero entry where $x_j=x_k=1, x_u=0\ \forall u\neq j,k$. Therefore, the entry in $\ket{\psi_{0}}$ will be given by the presence or absence of $x_jx_k,x_j,x_k$, in the ANF of the parent graph. Since ${\cal{A}}_j^T$ and ${\cal{A}}_k^T$  anti-commute, the restriction of ${\cal{A}}_j,{\cal{A}}_k$ (that  also anti-commute) to the pair $j,k$ will give $X_\alpha Z_\beta$ and $I_\alpha X_\beta$ for $\alpha,\beta\in\{j,k\}$. This implies that the restriction of ${\cal{A}}_\alpha^T$ and ${\cal{A}}_\beta^T$ are , respectively, $X_\alpha I_\beta$ and $Z_\alpha X_\beta$, so the restriction of $i{\cal{A}}_\alpha^T{\cal{A}}_\beta^T$ is $Y_\alpha X_\beta$. Note that we only need to look at the restriction, since any further Kronecker product will be by $X,\,Z$ or $I$, and these will not change the first non-zero entry.
\begin{itemize}
\item If the extension column gives $({\cal{A}}^e)_{\alpha,n}=Z$ and $({\cal{A}}^e)_{\beta,n}=Y$, we get the term $2(x_\alpha x_\beta+x_\beta)+x_\beta$ in the parent graph, so $\ket{\psi_{0}}$ has entry $i$ for $x_j=x_k=1, x_u=0\ \forall u\neq j,k$, which implies that $\bra{\psi_{0}}$ has entry $-i$, same as the same entry in the first row of the matrix given by the order of multiplication ${\cal{A}}^T_\alpha{\cal{A}}^T_\beta$. %, by which we mean that we multiply on the left the matrix $\sigma_a$ that has $({\cal{A}}^e)_{a,n}=Z$, and on the right the matrix $\sigma_b$ that has $({\cal{A}}^e)_{b,n}=Y$.
\item If the extension column gives $({\cal{A}}^e)_{\alpha,n}=Y$ and $({\cal{A}}^e)_{\beta,n}=Z$, we get the term $2x_\alpha+x_\alpha$ in the parent graph, so $\ket{\psi_{0}}$ has entry $-i$ (so $\bra{\psi_{0}}$ has entry $i$) for $x_j=x_k=1, x_u=0\ \forall u\neq j,k$,% which implies that $\bra{\psi_{0}}$ has entry $i$, for $x_j=x_k=1, x_u=0\ \forall u\neq j,k$, 
same as the same entry in the first row of the matrix given by the order of multiplication ${\cal{A}}^T_\beta{\cal{A}}^T_\alpha$.
\end{itemize}

\item Suppose ${\cal{A}}_j^T$ and ${\cal{A}}_k^T$ commute. Then, the term $\pm \frac{1}{2^{n}}{\cal{A}}^T_j{\cal{A}}^T_k\,0\leq j,k\leq n-1$ is present in the sum, the matrices are commuting since ${\cal{A}}^T_j{\cal{A}}^T_k$ is Hermitian. By inspection, on the pair  ${\cal{A}}_j,{\cal{A}}_k\in S$, $X_\alpha Z_\beta$ and $Z_\alpha X_\beta$ give entry -1 regardless of extension (note that they have the commuting entry in the extension column), same as ${\cal{A}}^T_j{\cal{A}}^T_k$. As for $X_\alpha I_\beta$ and $I_\alpha X_\beta$ give entry +1 regardless of extension, same as ${\cal{A}}^T_j{\cal{A}}^T_k$.
\item  By the proof of theorem \ref{pauli sum}, ${\cal{A}}_K^T, \,K\subseteq \{0,\ldots,n-1\}$, is present in the sum iff the corresponding stabilizer, ${\cal{A}}_K^T$,% \,A\subseteq \{0,\ldots,n-1\}$, 
has either $X$ or $I$ in the extension column. Therefore, there exists a (not necessarely unique) decomposition in size 1 and size 2 terms that have either $X$ or $I$ in the extension column (since $YZ=iX$), and, as the entry $x_j=1\ \forall j\in A, x_j=0 \ \forall j\notin A$ will depend on the sum of the terms for the size 1 and size 2 cases, as the Boolean function has degree at most 2, this entry will depend on the entry of the size 1 and size 2 cases, and therefore the coefficient will be the multiplication of the size 1 and size 2 cases (for any given decomposition).
\end{itemize}
\end{itemize}
{\em General $e$:} Since the case $e=1$ was independent of the actual stabilizers (it only depends on the extension), each new column will modify the Boolean expression accordingly. Therefore, the coefficient of any term is given by the multiplication of the coefficient resultant of each extension column.

{\em Change of sign:} Let us first assume  that $j$ has weight 1, with support in $k$. Note that each ${\cal{A}}_j^T$ has support in $(0,2^k)$. In terms of $\bra{\psi_{0\ldots0}}$, this means that changing the sign of ${\cal{A}}_j^T$ changes the sign in the $j$th element in $\ket{\psi_{0\ldots0}}$ (and therefore in $\bra{\psi_{0\ldots0}}$). This is equivalent to adding  a linear term $x_k$ as long as we change the sign all elements in$\ket{\psi_{0\ldots0}}$ (and therefore in $\bra{\psi_{0\ldots0}}$) with $x_k=1$. %Note that adding a Boolean linear term $x_k$ to a parent that does not have $k$ as a red node gives a parent with a different density matrix, while adding a Boolean linear term $x_k$ to a parent that does have $k$ as a red node gives a parent with the same  density matrix. 
Suppose now that $j$ has weight $>1$. If we  change the sign in $j$, we have to change also the signs of an odd number of its decomposing terms. Consider therefore the smallest terms present in the group. If the sign changed is of weight 1, see above. If the weight is $t>1$, with support in $k_1,\ldots,k_t$, then any of the ${\cal{A}}_{k_i}$ are not elements in the subgroup. Adding any linear term $x_{k_i}$ will give the desired matrix, since it is zero in the places were it might differ. %Therefore, it is also related by linear terms.
%If we don't change all the elements, we are not adding linear but higher degree Boolean terms, and the result matrix will not %be

%This concludes the proof of theorem \ref{pauli sum}.
%\end{itemize}
\end{pf}}

\begin{cor} \label{maximal}The commutative subgroup corresponding to a  child is maximal.%, where $e$ is the minimal number of EPR we have to add.
\end{cor}
\begin{pf}
Each column added gives a affine linear Boolean function $f_i$, which gives a constrain to the indicator $J$ of $f=\prod f_i$. Each constraint, $f_i$, is nontrivial (that is, not a constant), because if this were the case, then the column would only have $X$ and $I$, and would therefore be superfluous, yielding a contradiction. The $f_i$ are all independent, otherwise we get redundant columns, yielding a contradiction with $e$ being minimal. Furthermore, any new independent linear constraint reduces by half the size of the indicator, wich means that this size is equal to $2^{n-e}$, so it is a maximal commutative subgroup.
\end{pf}

NB: Note that allowing superflous constraints $f_i$ will give  {\em  commutative subgroups that are in general not maximal}. In this way, we could also extend graph states in a natural way: the density matrix for the graph state is given by the sum of all the elements of the stabilizer (since it is self-dual), and of course any consistent sign changes that give linear terms (eigenvalue -1). We can however define more density matrices that are stabilized by the stabilizer of the graph state, by allowing also smaller commutative subgroups. For instance, a density matrix associated to the undirected line from 0 to 1 could be then $\rho=a_0\left(I\otimes I\pm X\otimes Z\pm Z\otimes X\pm Y\otimes Y\right)+a_1\left(I\otimes I\pm X\otimes Z\right)+a_2\left(I\otimes I\pm Z\otimes X\right)+a_3\left(I\otimes I\pm Y\otimes Y\right)+a_4 I\otimes I$, where $\sum a_i=1$. Note that in the first term, the signs have to be consistent, changing for instance $X\otimes Z$ to $-X\otimes Z$ forces the change $Y\otimes Y$ to $-Y\otimes Y$.

\section{Open problem: the weighted sum of a maximal commutative subgroup of $S^\perp$ is a density matrix}

\begin{conj}
Any maximal commutative subgroup corresponds to a child of a pure parent graph, so its weighted sum (with appropiate coefficients) is a density matrix.
\end{conj}

\begin{proof}[incomplete]
Let $ M^\perp$ be a maximal  commutative subgroup of $S^\perp$  for a mixed graph G. Then, the corresponding elements in S commute as well (because direction of arrows is not important), and form a commutative subgroup, M. If there is an extension of S (not unique) to the stabilizer of to a pure graph state, the extension has $e$ columns. Since the elements of M commute, for $e=1$, we can assign $X$ for all odd products of elements of the basis $S$, and $I$ to even products. Since it is a subgroup, this is consistent. In general, we can at least always assign either $X$ or $I$ in the corresponding columns of the basis, which respects their commutativity. Not all combinations will be valid (for instance, assigning $I^{\otimes e}$ to a node that does not commute with everybody is not allowed) but there exist at least one such combination, by lemma \ref{3colour}. In general, we still have to prove that this combination is always possible.% (this is the unproven part).
%(here's the point I'm a bit uncertain about, because in principle we could have ``cancelling'' anticommutativities between the elements of M, but I think in that case we should be able %to substitute them with X or I's. Note cancelling wouldn't give me this density matrix, since anything other than X's is not in the final density matrix). 
If it exists, any such combination would have as density matrix the one generated by $M^\perp$.

A possible strategy for finding these extension columns is to write the indicator that will indicate the position of the $Y,Z$ positions. For any nondecomposable ${\cal{A}}_V^T$, we need to impose the condition that  in the indicator function, all the $x_j,\,j\in V$, are equal. This we can achieve by taking the products $\prod(x_j+k_k+1)$. This strategy works for non-intersecting elements. 
\begin{lem} Regardless of whether we get intersecting elements or not, we  always obtain $e$ terms in the final product of the indicator.\end{lem}
\begin{proof} Let $M$ be a maximal commutative subgroup. By theorem \ref{MaxCommSub}, $|M|=2^{n-e}$. Let $C$ be the corresponding binary linear code. A binary linear code with $k=n-e$ independent codewords has a parity-check matrix with $e$ rows: the product of any set of these $e$ parity-check conditions will be our indicator function.

\end{proof}
Note that since $e\leq\frac{n}{2}$ by definition, we have $2e\leq n$, which means that $e\leq n-e$. 
\end{proof}
\begin{ex}
Consider the mixed graph given by $$\left(\begin{array}{c|c}
X\otimes Z\otimes I\otimes I\otimes Z&={\cal{A}}_0\\
I\otimes X\otimes Z\otimes I\otimes Z&={\cal{A}}_1\\
I\otimes I\otimes X\otimes Z\otimes Z&={\cal{A}}_2\\
I\otimes I\otimes I\otimes X\otimes Z&={\cal{A}}_3\\
I\otimes I\otimes I\otimes I\otimes X&={\cal{A}}_4
\end{array}\right)$$ Then, the dual is $$\left(\begin{array}{c|c}
X\otimes I\otimes I\otimes I\otimes I&={\cal{A}}^T_0\\
Z\otimes X\otimes I\otimes I\otimes I&={\cal{A}}^T_1\\
I\otimes Z\otimes X\otimes I\otimes I&={\cal{A}}^T_2\\
I\otimes I\otimes Z\otimes X\otimes I&={\cal{A}}^T_3\\
Z\otimes Z\otimes Z\otimes Z\otimes X&={\cal{A}}^T_4
\end{array}\right)$$The rank of the corresponding adjacency matrix is 4, so $e=2$. Commuting with ${\cal{A}}_0^T$ is the subgroup generated by ${\cal{A}}_0^T,{\cal{A}}_2^T,{\cal{A}}_3^T,{\cal{A}}_{\{1,4\}}^T={\cal{A}}_4^T{\cal{A}}_1^T$. But not all of them commute so we choose another element in this subgroup. Commuting (e.g.) with both ${\cal{A}}_0^T$ and ${\cal{A}}_3^T$ is the subgroup generated by ${\cal{A}}_0^T,{\cal{A}}_3^T,{\cal{A}}^T_{\{1,2,4\}}$.  This is a maximal commutative subgroup. We shall find an extension of the stabilizers such that we obtain a parent pure graph state. 

For any nondecomposable ${\cal{A}}_D^T$, we need to impose the condition that  in the indicator function, all the $x_j,\,j\in D$, are equal. This we can achieve by taking the products $\prod(x_j+k_k+1)$. Here we have the composite ${\cal{A}}^T_{\{1,2,4\}}$, which gives the condition $(x_1+x_4+1)(x_2+x_4+1)$ (or any of the equivalent conditions). We assign then in the first column of the extension the entries $Z,Y$ to positions $1,4$, while the remaining entries of the first column will be $X$ or $I$. Since rows 1 and 4 anti-commute, we assign $Z$ to position 1 and $Y$ to position 4 (note that the other noncommutative choice will give the same commutative subgroup, with a change of sign in ${\cal{A}}^T_{\{1,2,4\}}$). Both ${\cal{A}}^T_0$ and ${\cal{A}}^T_2$ anti-commute with ${\cal{A}}^T_1$ and ${\cal{A}}^T_4$, so we assign $X$ to positions 0 and 2. Finally, we assign $I$ to position 3. Now for the second column: Here we assign $Z$ or $Y$ to positions 2 and 4, while we assign $X$ or $I$ to the remaining nodes.  Since ${\cal{A}}'^T_2$ and ${\cal{A}}'^T_4$ commute, we give them the same entry, say $Z$. Now, ${\cal{A}}'^T_0$ is already commutative with all the other rows, so we assign $I$ to position 0, and $I$ to position 1 for the same reason. We assign $X$ to position 3, since ${\cal{A}}'^T_3$ anti-commutes with both ${\cal{A}}'^T_2$ and ${\cal{A}}'^T_4$. This is now a fully commuting set:
$$\left(\begin{array}{c|cc}
X\otimes Z\otimes I\otimes I\otimes Z&X&I\\
I\otimes X\otimes Z\otimes I\otimes Z&Z&I\\
I\otimes I\otimes X\otimes Z\otimes Z&X&Z\\
I\otimes I\otimes I\otimes X\otimes Z&I&X\\
I\otimes I\otimes I\otimes I\otimes X&Y&Z
\end{array}\right)$$

\end{ex} 
\begin{ex}
Let $G$ be the arrrowed 6-clique given by the following stabilizer basis: 
$$\left(\begin{array}{cccccc}
X&Z&Z&Z&Z&Z\\
I&X&Z&Z&Z&Z\\
I&I&X&Z&Z&Z\\
I&I&I&X&Z&Z\\
I&I&I&I&X&Z\\
I&I&I&I&I&X
\end{array}\right)$$ Then, $e=3$. A maximal commutative subgroup is generated by ${\cal{A}}^T_{0},{\cal{A}}^T_{\{1,2\}},{\cal{A}}^T_{\{3,4\}}$. This gives the indicator $(x_1+x_2+1)(x_3+x_4+1)(x_5+1)$. According to this, we assign in the first extension column the entries $Z$ and $Y$ to positions 1 and 2, respectively since their respective rows anti-commute. The remaining entries on this column have to be either $X$ or $I$. We assign $X$ to positions 0,3,4,5, because their respective rows anti-commute with both. Now for the second column: We assign $Z$ and $Y$ to positions 3 and 4, respectively since their respective rows anti-commute. The remaining entries on this column have to be either $X$ or $I$. We assign $X$ to positions 0 and 5, because their respective rows (still) anti-commute with both, $I$ to positions 1 and 2, which respective rows now commute with all rows. Finally, in the last column we assign $Z$ to position 5. The remaining entries on this column have to be either $X$ or $I$. We assign $X$ to position 0, and $I$ to all the other positions, since their rows all commute pairwise now. This is now a fully commuting set.
$$\left(\begin{array}{cccccc|ccc}
X&Z&Z&Z&Z&Z&X&X&X\\
I&X&Z&Z&Z&Z&Z&I&I\\
I&I&X&Z&Z&Z&Y&I&I\\
I&I&I&X&Z&Z&X&Z&I\\
I&I&I&I&X&Z&X&Y&I\\
I&I&I&I&I&X&X&X&Z
\end{array}\right)$$
\end{ex}

\section{Appendix A}\label{AppA}
%\subsection{Definition of the coefficients in $S^T$}
%\vspace{2mm}

We re-iterate that, as with $S$, if some or all the rows of ${\cal A}$ pairwise anti-commute, then some or all of the members of $S^{\perp}$ are defined up
to a global constant of $\pm 1$. This does not affect the commutation relations between members of $S$ and $S^{\perp}$. But which members of $S^{\perp}$ are defined up to $\pm 1$ and how many of them? Consider the $n$-vertex mixed graph ${\overleftarrow{G}}$ with adjacency matrix $A^T$, derived from ${\cal A}^T$. Then $G_b$ is also the undirected graph associated with ${\overleftarrow{G}}$ (as well as with
$G$), as $G_b$ has adjacency matrix $\Gamma = A + A^T$. Then the members of $S^{\perp}$ associated with a $+1$ coefficient are obtained as the subset of row products of ${\cal A}^T$ indexed by sets representing all possible independent sets in $G_b$. But how to compute these sets? For a set $w$, the {\em power set} of $w$, ${\cal P}(w)$ is the set of all subsets of $w$ including the empty set and $w$ itself.
Define $V$ to be a family of subsets of $\Z_n$, specifically let $V$ represent the maximum independent sets in $G_b$ meaning that, for $v \subset Z_n = \{0,1,\ldots,n-1\}$, $v \in V$ iff $v$ is an independent set in $G_b$ and there does not exist a $v'$ which is an independent set in $G_b$ such that $v \subset v'$, $v \ne v'$. Then $G_b$ specifies $V$ precisely and $V$ specifies $G_b$ precisely. For instance, for $n = 6$ let $G_b$ be defined by the adjacency matrix $\Gamma =
\left ( \begin{array}{cccccc}
0 & 0 & 0 & 1 & 1 & 1 \\
0 & 0 & 0 & 0 & 1 & 0 \\
0 & 0 & 0 & 0 & 0 & 0 \\
1 & 0 & 0 & 0 & 0 & 0 \\
1 & 1 & 0 & 0 & 0 & 1 \\
1 & 0 & 0 & 0 & 1 & 0
\end{array} \right )$. Then $V = \{\{0,1,2\},\{2,3,4\},\{1,2,3,5\}\}$. Observe that $V$ also lists the maximum cliques in the graph complement, ${\overline{G_b}}$ of $G_b$. From $V$ we can compute those row subsets of ${\cal A}^T$ whose product is independent of product order, i.e. the row subsets that are fully commutative. We call this set of subsets ${\cal{E}}(V)$ and a member of ${\cal{E}}(V)$ represents a unique member of $S^{\perp}$ with a $+1$ coefficient.
For our example there are $|{\cal{E}}(V)| = 24$ such row subsets, namely
$$ \begin{array}{ll} {\cal{E}}(V) = & \{\emptyset,\{0\},\{1\},\{0,1\},\{2\},\{0,2\},\{1,2\},\{0,1,2\}, \\
 & \{3\},\{1,3\},\{2,3\},\{1,2,3\},\{4\},\{2,4\},\{3,4\},\{2,3,4\}, \\
 & \{5\},\{1,5\},\{2,5\},\{1,2,5\},\{3,5\},\{1,3,5\},\{2,3,5\},\{1,2,3,5\}\}.
\end{array} $$
Then ${\overline{{\cal{E}}(V)}} = {\cal P}(\Z_n) \setminus {\cal{E}}(V)$ is the family of
subsets of $\Z_n$ that represent non-commuting row subsets of ${\cal A}^T$, i.e. a member
of ${\overline{{\cal{E}}(V)}}$ represents a unique member of $S^{\perp}$ with a $\pm 1$
coefficient. The number of them is $|{\overline{{\cal{E}}(V)}}| = 2^n - |{\cal{E}}(V)|$, this being the number of members of $S^{\perp}$ with a $\pm 1$ coefficient. For our example $|{\overline{{\cal{E}}(V)}}| = 2^6 - 24 = 40$. We can express ${\cal{E}}$ in terms of
power sets as follows,
$$ {\cal{E}}(V) = \bigcup_{w \in V} {\cal P}(w). $$

We now provide a recursive algorithm to compute ${\cal{E}}(V)$, given
${\cal A}^T$ (and therefore $V$). Let $w \subset \Z_n$,
$\hat{V} = \bigcup_{w,w' \in V, w \ne w'} w \cap w'$ and, for $w \in V$, let
$V_w = \bigcup_{w' \in V \setminus w} w \cap w'$. For our example, $V = \{\{0,1,2\},\{2,3,4\},\{1,2,3,5\}\}$,
$\hat{V} = \{\{0,1,2\} \cap \{2,3,4\},\{0,1,2\} \cap \{1,2,3,5\}, \{2,3,4\} \cap \{1,2,3,5\}\} = \{\{2\},\{1,2\},\{2,3\}\}$
and $V_{\{0,1,2\}} = \{\{0,1,2\} \cap \{2,3,4\},\{0,1,2\} \cap \{1,2,3,5\}\} = \{\{2\},\{1,2\}\}$.
Then, recursively,
$$ {\cal{E}}(V) = {\cal{E}}(\hat{V}) \cup \bigcup_{w \in V} {\cal P}(w) \setminus {\cal{E}}(V_w). $$
%where we can simplify the family of sets as we go along, as follows. For $w \in \Z_n$ and $V$ a set family in $\Z_n$, where $w \subseteq v$ for one or more $v \in V$, then the set family
%$V' = w \cup V$ can be simplified to $V' = V$. Moreover
%${\cal{E}}(\{\{w\}\}) = {\cal P}(\{w\})$ and, for $V = \{\{w\}\}$, $\hat{V} = \emptyset$.
So, for our example with
$V = \{\{0,1,2\},\{2,3,4\},\{1,2,3,5\}\}$ we obtain
$$ \begin{array}{rl}
{\cal{E}}(\hat{V}) = & {\cal{E}}(\{\{2\},\{1,2\},\{2,3\}\}) \cup {\cal P}(\{0,1,2\}) \setminus {\cal{E}}(\{\{2\},\{1,2\}\}) \\
 & \cup \mz {\cal P}(\{2,3,4\}) \setminus {\cal{E}}(\{\{2\},\{2,3\}\}) \cup
{\cal P}(\{1,2,3,5\}) \setminus {\cal{E}}(\{\{1,2\},\{2,3\}\}) \\
 = & {\cal{E}}(\{\{1,2\},\{2,3\}\}) \cup {\cal P}(\{0,1,2\}) \setminus {\cal P}(\{1,2\}) \\
 &  \cup \mz {\cal P}(\{2,3,4\}) \setminus {\cal P}(\{2,3\}) \cup
{\cal P}(\{1,2,3,5\}) \setminus {\cal{E}}(\{\{1,2\},\{2,3\}\}) \\
 = & \{\{0\},\{0,1\},\{0,2\},\{0,1,2\}\} \cup \{\{4\},\{2,4\},\{3,4\},\{2,3,4\}\}
 \cup {\cal P}(\{1,2,3,5\}),
\end{array} $$
from which we obtain the family of 24 sets written above.

%\vspace{3mm}
\section{Appendix B}\label{AppB}
\subsection{$e = 1$ qubit extensions}
We consider here the simplest case where $e = \frac{1}{2}{\m{rank}}(\Gamma) = 1$.

\begin{lem}
Let $j$ be a vertex of the mixed graph described by $A$, for which $e = 1$, such that vertex
$j$ is either disconnected or
connected to other vertices only via undirected edges. Then ${\cal A}^e_{j,n} = I$.
\label{ignoreundir}
\end{lem}
\begin{pf}
From the conditions in the lemma, row $j$ of ${\cal A}$ commutes with all other rows in ${\cal A}$.
As $e = 1$ then there are at least two other rows of ${\cal A}$, rows $h$ and $k$, such that
${\cal A}^e_{h,n} \ne {\cal A}^e_{k,n}$, ${\cal A}^e_{h,n},{\cal A}^e_{k,n} \in \{X,Y,Z\}$. As we
require row $j$ of ${\cal A}^e$ to commute with rows $h$ and $k$ of ${\cal A}^e$ and, as row $j$ of
${\cal A}$ commutes with rows $h$ and $k$ of ${\cal A}$, then the only choice for ${\cal A}^e_{j,n}$ is $I$.
\end{pf}

\vspace{2mm}

A consequence of Lemma \ref{ignoreundir} is that, given a mixed graph described by $A$, we may immediately,
and wlog, set all ${\cal A}^e_{j,n} = I$ for unconnected and undirected vertices $j$ and focus on solving the
extension problem for the residual directed graph comprising vertices involved in one or more directed
edges. This is implicitly the case for Theorem \ref{graphparent} below.

\begin{thm}
Let ${\cal A}$ be an $n \times n$ Pauli stabilizer matrix such that, from (\ref{dirrank}), $e = 1$.
Then ${\cal A}$ can always be extended to a fully-commuting stabilizer matrix, ${\cal A}^e$
of the form ${\cal A}^e = i^{I_{n+e} \star A^e}X^{I_{n+e}}Z^{A^e}$. Such a matrix stabilizes a graph
state of $n + e = n + 1$ qubits, and has elements from $\{X,Y\}$ on the diagonal and
elements from $\{I,Z\}$ off it.
\label{graphparent}
\end{thm}
\begin{pf} (of Theorem \ref{graphparent})
\begin{lem}
For ${\cal A}^e = \left ( {\cal A}^e_{i,j} \in \{I,X,Z,Y\}, 0 \le i,j \le n \right )$, we can wlog
fix ${\cal A}^e_{n,n} = X$.
\label{Xcorner}
\end{lem}
\par {\addtolength{\leftskip}{5mm}
\begin{pf} (of Lemma \ref{Xcorner})
If ${\cal A}^e_{n,n} = X$ then we are done. If ${\cal A}^e_{n,n} \in \{Z,Y\}$ then,
by Lemma \ref{ConjP}, we can conjugate
${\cal A}^e_{n,n}$ to $X$ and, by Lemma \ref{envuni}, this leaves the density matrix stabilized
by ${\cal A}$ unchanged, and preserves the pairwise commutation properties.
If  ${\cal A}^e_{n,n} = I$ then we multiply row $n$ by any other row, $j$, where
${\cal A}^e_{j,n} \in \{X,Y,Z\}$, i.e.
${\cal A}^e_{n,n} \leftarrow {\cal A}^e_{j,n}{\cal A}^e_{n,n}$. At least one such row, $j$, must
exist for $e \ge 1$. This operation preserves the codespace of ${\cal A}^e$, and results in
${\cal A}^e_{n,n} \in \{X,Y,Z\}$. Then we continue with conjugation if necessary, as discussed
previously.
\end{pf}

Although Lemma \ref{Xcorner} restricts ${\cal A}^e_{n,n} = X$ we emphasise that there is nothing
special about $X$ here, we could've chosen to conjugate to $Y$ or $Z$. But fixing to $X$
simplifies the problem for us without costing anything in terms of generality, as,
by Lemma \ref{envuni}, the associated density matrix is unchanged.
\par }
\begin{lem}
If the rows $i$ and $j$ of ${\cal A}$ are pairwise non-commuting then \\
$$ {\cal A}^e_{i,n}, {\cal A}^e_{j,n} \in \{X,Z,Y\}, \mf {\cal A}^e_{i,n} \ne {\cal A}^e_{j,n}. $$
\label{3colour}
\end{lem}
\par {\addtolength{\leftskip}{5mm}
\begin{pf} (of Lemma \ref{3colour})
If rows $i$ and $j$ of ${\cal A}$, $i \ne j$, do not pairwise commute then,
if ${\cal A}^e_{i,n} = {\cal A}^e_{j,n}$, or if ${\cal A}^e_{i,n}$ and/or ${\cal A}^e_{j,n}$ equals
$I$, then rows $i$ and $j$ of ${\cal A}^e$ do not pairwise commute either. But we
require ${\cal A}^e$ to be fully pairwise commuting. So ${\cal A}^e_{i,n} \ne {\cal A}^e_{j,n}$ and
${\cal A}^e_{i,n}, {\cal A}^e_{j,n} \in \{X,Z,Y\}$.
\end{pf}
\par }
\begin{lem}
If the rows $i$ and $j$ of ${\cal A}$ are pairwise commuting then \\
$$ {\cal A}^e_{i,n} = {\cal A}^e_{j,n} \mf \m{ or } \mf \m{at least one of
${\cal A}^e_{i,n}$ or ${\cal A}^e_{j,n}$ equals $I$}. $$
\label{comm}
\end{lem}
\par {\addtolength{\leftskip}{5mm}
\begin{pf} (of Lemma \ref{comm})
Rows  $i$ and $j$ of ${\cal A}^e$ must pairwise commute and, if rows $i$ and $j$ of ${\cal A}$
pairwise commute, then matrices
${\cal A}^e_{i,n}$ and ${\cal A}^e_{j,n}$ must also pairwise commute, hence the conditions of the
lemma.
\end{pf}
\par }
\begin{lem}
Wlog,
the elements ${\cal A}^e_{n,j}$, $0 \le j < n$ can be restricted to elements from $\{I,Z\}$.
\label{IZRestrict}
\end{lem}
\par {\addtolength{\leftskip}{5mm}
\begin{pf} (of Lemma \ref{IZRestrict})
The last row, ${\cal A}^e_{n} = ({\cal A}^e_{n,j}, 0 \le j \le n)$, of
${\cal A}^e$, must pairwise commute with each of the $n$ previous rows. We have, wlog, already fixed
${\cal A}^e_{n,n} = X$. We choose elements of ${\cal A}^e_{n,j}$, $0 \le j < n$ from
$\{I,X,Z,Y\}$. For each ${\cal A}^e_{n,j} \in \{X,Y\}$ we successively replace the
$n$th row of ${\cal A}^e$, i.e. ${\cal A}^e_{n}$, with the product of rows $j$ and $n$, i.e.
${\cal A}^e_{n} \rightarrow {\cal A}^e_{j}{\cal A}^e_{n}$ which ensures that, now,
${\cal A}^e_{n,j} \in \{I,Z\}$, and continue in this fashion until
${\cal A}^e_{n,j} \in \{I,Z\}$, $0 \le j < n$. This process preserves the pairwise commuting property
for the rows of ${\cal A}^e$.
\end{pf}
\par }
The row product operations of lemma \ref{IZRestrict} may change ${\cal A}^e_{n,n} = X$ to
${\cal A}^e_{n,n} \in \{Z,Y\}$ in which case, using Lemma \ref{Xcorner},
a further conjugation operation on the $n$th column
of ${\cal A}^e$ is required to change ${\cal A}^e_{n,n}$ back to $X$. Note that, by Lemma
\ref{envuni}, such an operation leaves unchanged the density matrix, $\rho$, described by
${\cal A}$.

At this point ${\cal A}^e$ has been transformed, by row operations and
conjugation of the last column, to a matrix with $X$ or $Y$ on the diagonal, and with $\{I,Z\}$ off
the diagonal, apart from the last column where ${\cal A}^e_{j,n} \in \{I,X,Y,Z\}$, $0 \le j < n$, and
${\cal A}^e_{j,n} = X$. If ${\cal A}^e_{j,n} \in \{X,Y\}$ then we replace the $j$th row thus
${\cal A}^e_{j} \rightarrow {\cal A}^e_{n}{\cal A}^e_{j}$ so that, now,
${\cal A}^e_{j,n} \in \{I,Z\}$.

At this point ${\cal A}^e$ has been transformed, by row operations and
conjugation of the last column, to a matrix with $X$ or $Y$ on the diagonal, and with $\{I,Z\}$ off
the diagonal.

\begin{lem}
${\cal A}^e$ is of the form
${\cal A}^e = i^{I_{n+e} \star A^e}X^{I_{n+e}}Z^{A^e}$,
so must be a symmetric matrix.
\label{AeSym}
\end{lem}
\par {\addtolength{\leftskip}{5mm}
\begin{pf} (of Lemma \ref{AeSym})
All rows of ${\cal A}^e$ pairwise commute. Given that ${\cal A}^e$ has $X$ or $Y$ on the diagonal and
$\{I,Z\}$ off the diagonal, then this is only possible if ${\cal A}^e$ is symmetric.
\end{pf}
\par }

Combining Lemmas \ref{Xcorner}, \ref{3colour}, \ref{IZRestrict}, and \ref{AeSym}, one arrives at
the proof of Theorem \ref{graphparent}.
\end{pf}

\subsection{Conditions for $G$ to have mixed rank $e=1$}
Remember that $G$ is the mixed graph defined by adjacency matrix $A$, and $G_b$ the undirected graph with adjacency matrix $\Gamma=A+A^T$.

\begin{lem}\label{completetripartite}$G$ has mixed rank $e=1$ $\Leftrightarrow $ $G_b$ is a complete tripartite graph\footnote{Thanks to Jon Eivind Vatne for useful discussions about graph theory.}, with  possible isolated vertices, and where one of the partitions can be empty (so that $G_b$ is actually a complete bipartite graph).
\end{lem}

\begin{pf} If $(i,j)$ is an edge in $G_b$, then rows $i$ and $j$ of the stabilizer basis are anti-commuting so, by lemma \ref{3colour}, 
the corresponding elements in the extension column of ${\cal A}$, $({\cal A}^e_{i,n}, {\cal A}^e_{j,n})$ can only take the pairs $(X,Z)$, $(X,Y)$, $(Z,X)$, $(Z,Y)$, $(Y,X)$, and
$(Y,Z)$. In other words the edge $(i, j)$ in $G_b$ forces vertices $i$ and $j$ to be associated
with different operators in the extension column ${\cal A}^e_{n}$ - one can think of $X,Z,Y$ as three colours. Therefore $e=1$ is only possible if $G_b$ is three-colourable, which implies that the graph is empty (1-coloured), bipartite (2-coloured) or tripartite (3-coloured). If $G_b=\emptyset$ then $rank(\Gamma)=0$, so we discard this option. Isolated vertices correspond to the rows that commute with all other rows, and can be extended by the column entry $I$, as seen in lemma \ref{ignoreundir}. Let us now consider the connected part of $G_b$:

We are now going to prove that, if $G$ has mixed rank $e=1$, then the connected part of  $G_b$ is  complete tripartite or complete bipartite, with possible isolated vertices; since it has to be 2- or 3-coloured, we have to prove that all the vertices on each partition are connected to all the vertices in the other partition(s).
%Suppose this is not true. Then there is at least one pair of vertices $(u,v)$ such that $u$ and $v$ are in different partitions, and there is no edge between them. The partitions have assigned two different entries for the extension column, in the set $\{1,\omega,\omega^2\}$. \footnote{Not 0 because this is assigned only to the rows that connect with any other rows, and we are considering here the connected part of $G_b$.} But that means that the rows $u$ and $v$ anti-commute, since the entries on the extension column anti-commute. This gives that there must be an edge between $u$ and $v$ in $G_b$, which is a contradiction.
Consider any pair of vertices $(u,v)$ such that $u$ and $v$ are in different partitions. The partitions are assigned two different colours/entries for the extension column, in the set $\{Z,X,Y\}$. \footnote{Note that I is not present because since it does not change the commutativity it can only be assigned to the rows that connect with any other rows (which correspond to isolated nodes on $G_b$), and we are considering here the connected part of $G_b$.} Therefore the rows $u$ and $v$ anti-commute, since the entries in the extension column anti-commute. Therefore there must be an edge between $u$ and $v$ in $G_b$, which proves that the graph is complete bipartite or complete tripartite.

We will now prove that if $G_b$ is a complete bipartite or tripartite graph with possible isolated vertices, then $G$ has mixed rank $e=1$:

In general, the extension of $G$ depends only on $G_b$: We have that $(u,v)\in G_b$ iff row $u$ anti-commutes with row $v$. This is true because  all unitaries $U_j$ that form the stabilizer basis for a mixed graph state have exactly one $X$ in position $j$, and $Z$ and $I$ in the other positions. If $(u,v)\in G_b$, we have in rows $u$ and $v$ the pair $X_uZ_v,I_uX_v$ ($u\rightarrow v$) or $X_uI_v,Z_uX_v$ ($u\leftarrow v$), and since all entries other than $u$ and $v$ commute (the set $\{Z,I\}$ is commutative), then row $u$ anti-commutes with row $v$. On the other hand, if row $u$ anti-commutes with row $v$, then one of the pairs $X_uZ_v,I_uX_v$ or $X_uI_v,Z_uX_v$   must be present in rows $u$ and $v$, because as discussed all entries other than $u$ and $v$ commute, and these pairs correspond to $u\rightarrow v$ and $u\leftarrow v$, so $(u,v)\in G_b$. This means that the choice of the extension column/row will depend solely on $G_b$, since the only factor is the commutativity/anti-commutativity of the rows.

The extension of $G$ by a node  depends on whether we can extend the matrix by one column. We assign to each of the partitions the colours $Z,X,Y$. Since the graph is complete bipartite or tripartite, then any two nodes in the same partition get assigned the same colour/entry, and any two nodes in a different partition get assigned a different colour/entry. In particular, this means that there are no arrowed edges between any two nodes in the same partition, which means that their rows commute. Since they have the same extension entry, their extended rows also commute. This also means that there is an arrowed edge between any two nodes in different partitions, which means that their respective rows anti-commute. Since the extension entry is different for each of these two nodes, and is in the set $\{Z,X,Y\}$, this means that the extension rows now commute. This proves that one extension column is enough.

%Finally, adding isolated vertices does not change the fact that we can extend by one column, since being isolated on $G_b$ implies that they have no arrowed edges on $G$; therefore, their respective rows have to commute with any other rows, and we assign to them the entry $I$ in the extension column, as seen in lemma \ref{ignoreundir}.
\end{pf}

%\begin{lem}
%If $G_b$ is such that $\m{PAR}_H((-1)^b) = 2^{n-2}$, then $G_b$ is a three-colourable graph.
%\label{threecolour}
%\end{lem}
%\begin{pf}
%If $(ij)$ is an edge in $G_b$, then rows $i$ and $j$ of $A_G$ are anti-commuting, so
%$(c_i,c_j)$ can only take the pairs $(X,Z)$, $(X,Y)$, $(Z,X)$, $(Z,Y)$, $(Y,X)$, and
%$(Y,Z)$. In other words the edge $(ij)$ in $G_b$ forces vertices $i$ and $j$ to be associated
%with different operators in $C$ - one can think of $X,Z,Y$ as three colours. Therefore a solution
%for $C$ is only possible if $G_b$ is three-colourable. But we know, by corollary
%\ref{onecolumn}, that a one-column solution for $C$ is only possible if $\m{PAR}_H((-1)^b) = 2^{n-2}$.
%\end{pf}

%\section{Mixed graph states for general $e$}
%In order to be deal with those ${\cal A}$ where $e > 1$ it is convenient to consider a recursive
%`greedy' approach. For instance, for $e = 2$ we must extend ${\cal A}$ by two columns and
%rows so as to obtain ${\cal A}^e$. The first column/row extension must decrement the mixed
%rank, $e$, by 1, otherwise the second column row/extension would have to decrement $e$ by 2, which
%is impossible. Thus we can consider the process in two recursive steps, where each step reduces
%$e$ by 1 and produces 3 parents.
%
%\begin{lem}
%A mixed state described by ${\cal A}$ has $3^e$ parents, and can be written as
%$$ \rho = \sum_{i = 0}^{3^e - 1} c_i\rho_i, $$
%where $\sum_{i = 0}^{3^e - 1} c_i = 1$ and the $3^e$ density matrices $\rho_i$ are obtained from
%the $3^e$ parents, $\ket{\psi_{r,i}}$, by tracing out the $e$ environmental qubits.
%\end{lem}

\subsection{On parents and children for $e = 1$}

In general, as we shall now show, for $e = 1$ we obtain (up to local equivalence) a maximum of three mixed states, obtained from a maximum of three parents (up to local equivalence), $\ket{\psi_e}$, as stabilized by ${\cal A}^e$.
% Remember that $G$ is the mixed graph defined by adjacency matrix $A$, and $G_b$ the undirected graph that has adjacency matrix $\Gamma=A+A^T$.

\begin{lem}\label{6choices} A maximum of 6 distinct children is possible when $e = 1$.\end{lem}

\begin{pf} Isolated nodes in $G_b$ have the extension $I$, so we consider only non-isolated nodes. If $G_b$ is not connected then there are two or more non-empty subgraphs that are not connected to each other; however,  then $rank(\Gamma)>2$, which contradicts $e=1$. This is because $G_b$ is a nondirected graph for which any non-empty subgraph has rank at least 2, and the rank of $G_b$ is the sum of the ranks of the unconnected subgraphs.

%Empty nodes are ok in $G_b$, and the only choice for non-empty $G_b$ is a 0 (the Pauli matrix I) for these nodes, since they have to commute with everybody.

%So we only have to consider connected $G_b$ for the digraphs extended by only one column.

Lemma \ref{completetripartite} states that $G$ is at most 3-colourable if $e = 1$, as $G_b$ is complete tripartite (or complete bipartite). $G_b$ determines univocally whether two nodes have the same colour or not. There are 3 choices for the first colour, 2 for the second, and 1 for the remaining colour. In total that makes in principle 6 possible extension columns, yielding 6 different parents. Note that we only consider choices for ${\cal{A}}^e_{n,j}$, for $j=0,\ldots,n-1$. The reason for this is that that, as shown previously, we can fix ${\cal{A}}^e_{n,n}=X$.
\end{pf}
%However, a permutation $X\leftrightarrow  Y$ gives the same graph, but for a digit $a_n$ in the diagonal.  HOWEVER, THE CHOICE IS EXCLUDING a_n, SO THIS REDUCES JUST FROM 12 TO 6. The reason for this is that in binary form $(A|B)$ you have to add the last row to all rows having either $ X (\omega)$ or $Y (\omega^2=\omega+1)$ in the extended column. The permutation $\omega \leftrightarrow \omega+1$ will induce the choice $a_n\leftrightarrow  a_n +1$, leaving the rest of the graph invariant.  %So of the 12 parent graphs, 6 will be identical to the %other 6 except for the linear term $i^{x_n}$.
%This means that the choice of extension column reduce to half, that is, to 3.
%We can always add the appending linear term to the resulting parents to get the other 6 parents.

%The resulting parents are not necessarily distinct, depending on the automorphism group of the graph (any permutation in the automorphism group of the graph will give you an automorphism of the parents, for permutations of the choice of colours).
\begin{lem} \label{equiv} When $e=1$ we obtain at most 3 local unitary inequivalent children.
\end{lem}
\begin{pf}
Denote the resultant matrix after applying permutation $(Z\ Y)$ on  ${\cal{A}}^e_{n,j}$,  $j=0,\ldots,n-1$, as ${\cal{A}'}^e_{n,j}$. The last rows of ${\cal{A}}$ and ${\cal{A}}'$ are equal.
We can make both ${\cal{A}}$ and ${\cal{A}}'$ symmetric by multiplying by the last row the rows where the last column entry is $X$ or $Y$. %${\cal{A}}^e_{n,j}=X$ or ${\cal{A}}^e_{n,j}=Y$, we
Thus the rows where ${\cal{A}}^e_{n,i}=X$ yield the same result as the rows where ${\cal{A}'}^e_{n,i}=X$, while we multiply the last row (which is equal in both matrices) by the rows where ${\cal{A}}^e_{n,j}={\cal{A}'}^e_{n,k}=Y$. This means that ${\cal{A}}^e_{n,k}={\cal{A}'}^e_{n,j}=Z$, so they do not get multiplied in their respective matrices.  %In the parent graph we get that the new node $n$ is connected to all $j$ and $k$.
The effect of this is to complement the neighbourhood of node $n$, which comprises nodes $j$ and $k$, and to induce the permutation  $(X\ Y)$ in the diagonal position of these rows.

On the other hand, %in the parent graph we get that the new node $n$ is connected to all $j$ and $k$. I
if we perform %ocal Complementation (LC) on node $n$ in the parent corresponding to ${\cal{A}}$, which is given by
the local unitary operation $Z_{K}(DN)_n$, %D_{{\cal{N}}_n}$
where
$Z=\left(\begin{array}{cc}
  1&0\\
  0&-1
  \end{array}\right)$, $D=\left(\begin{array}{cc}
  1&0\\
  0&i
  \end{array}\right)$ and $N=\frac{1}{\sqrt{2}}\left(\begin{array}{cc}
  1&i\\
  1&-i
  \end{array}\right)$, and $K$ is the largest set such that ${\cal{A}}^e_{n,k}=Z\ \forall k\in K$, the effect is the same as described above (see \cite{thesis} for a description of the effect of $N$). This implies that of the 6 parents, only at most 3 are local unitary inequivalent.
\end{pf}

We therefore obtain the following result.

\begin{lem}
Let ${\cal A}$ be such that, from (\ref{dirrank}), $e = 1$. Then ${\cal A}$ represents the mixed
graph state
$$ \rho = \sum_{j=0}^5  c_j\rho_j, \mf \sum_{j=0}^5 c_j = 1, \mf c_j \ge 0, \forall j. $$
%c_0\rho_0 + c_1\rho_1 + c_2\rho_2, \mf c_0 + c_1 + c_2 = 1, c_j\geq0, $$
where
$\rho_j = \ketbra{\phi_{0,j}}{\phi_{0,j}} + \ketbra{\phi_{1,j}}{\phi_{1,j}}$, $0 \le j < 6$, are
the six density matrices obtained from the six parents
$\ket{\psi^e_{j}}$, as stabilized by ${\cal A}^e_{j}$.
\label{mix}
\end{lem}
\begin{pf} Follows immediately from lemma \ref{6choices}.
%The proof that   there are a maximum of 6 extensions that we need to consider is given in lemma \ref{6choices}.
\end{pf}

Although only at most 3 of them are locally inequivalent, we include them all in the general sum $\rho$, since changing $\rho_j$ for a  locally equivalent $\rho_j'$ does not necessarily give local equivalence in $\rho$.

\end{document}